\RequirePackage{fix-cm}

\documentclass[smallextended]{svjour3}       

\smartqed  

\usepackage{etex}
\usepackage{bbding}
\usepackage{dingbat}
\usepackage[ruled]{algorithm2e}
\usepackage{amsmath}
\usepackage{algorithmic}
\usepackage{amsfonts}
\usepackage{amsmath}
\usepackage{amssymb}
\usepackage{latexsym}
\usepackage{graphicx}
\usepackage{color}
\usepackage[normalem]{ulem}
\usepackage{balance}
\usepackage{nicefrac}
\usepackage{wasysym}
\usepackage{enumitem}
\usepackage{gensymb}
\usepackage{ctable}
\usepackage{url}
\usepackage{longtable}
\usepackage{citesort}

\usepackage{pstricks,pst-node,pst-plot,pst-coil,pst-text,multido,pst-3d,pst-tree,pst-grad,graphics,graphicx}

\usepackage{frcursive}

\usepackage{pifont,epsfig,rotating}

\usepackage{etoolbox}

\DeclareMathAlphabet{\mathpzc}{OT1}{pzc}{m}{it}

\definecolor{dgreyblue}{rgb}{0.26,0.3,0.46}             

\newcommand{\bee}{{\mathbf{e}}}
\newcommand{\cA}{\mathcal{A}}
\newcommand{\cB}{\mathcal{B}}
\newcommand{\cC}{\mathcal{C}}
\newcommand{\cE}{\mathcal{E}}
\newcommand{\cV}{\mathcal{V}}

\newcommand{\cL}{\mathcal{L}}

\renewcommand{\text}[1]{\hbox{\rm \ #1\ \/}}

\newcommand{\be}[1]{\begin{equation}\label{#1}}
\newcommand{\ee}{\end{equation}}
\newcommand{\beqn}{\begin{eqnarray*}}
\newcommand{\eeqn}{\end{eqnarray*}}
\newcommand{\beq}{\begin{eqnarray}}
\newcommand{\eeq}{\end{eqnarray}}
\newcommand{\ben}{\begin{enumerate}}
\newcommand{\een}{\end{enumerate}}
\newcommand{\bi}{\begin{itemize}}
\newcommand{\ei}{\end{itemize}}
\newcommand{\eps}{\varepsilon}

\newcommand{\IE}{{\em i.e.}\xspace}
\newcommand{\tx}{^{\rm th}}

\newtheorem{fact}{Fact}

\renewenvironment{proof}{{\noindent\bf Proof.\ }}{\hfill{\Pisymbol{pzd}{113}}\vspace{0.1in}}
\newenvironment{proof-sketch}{{\noindent\bf Sketch of Proof.\ }}{\hfill{\Pisymbol{pzd}{113}}\vspace{0.1in}}

\newcommand{\NP}{\mathsf{NP}}
\newcommand{\LP}{\mathsf{LP}}

\newcommand{\cF}{\mathcal{F}}
\newcommand{\cP}{\mathcal{P}}

\newcommand{\TB}{\vspace{-0.1ex}}\newcommand{\TiE}{\setlength{\itemsep}{-1ex}}

\newcommand{\comment}[1]{}

\newcommand{\EG}{{\it e.g.}\xspace}
\newcommand{\FI}[1]{Fig.~\ref{#1}\xspace}

\newcommand{\h}{{\mathfrak{h}}}

\newcommand{\cQ}{\mathcal{Q}}
\newcommand{\dist}{\mathrm{dist}}
\newcommand{\vol}{\mathsf{vol}}
\newcommand{\have}{\pmb{\overline{h}}}
\newcommand{\cut}{\mathrm{cut}}
\newcommand{\g}{\mathpzc{g}}

\newcommand{\OPT}{{ \mathsf{OPT} } }
\newcommand{\mse}{{\sc Mse}}
\newcommand{\uumv}{{\sc Uumv}}
\newcommand{\hssc}{{\sc Hssc}}
\newcommand{\ehssc}{{\sc Ehssc}}

\newcommand{\sse}{{\sc Sse}}

\definecolor{columbiablue}{rgb}{0.61, 0.87, 1.0}

\allowdisplaybreaks

\journalname{Algorithmica}

\begin{document}

\title{Effect of Gromov-hyperbolicity Parameter on Cuts and Expansions in Graphs and Some Algorithmic Implications}

\titlerunning{Effect of Gromov-hyperbolicity on Cuts and Expansions}        

\author{
Bhaskar DasGupta
\and
Marek Karpinski
\and
Nasim Mobasheri
\and
Farzane Yahyanejad
}

\institute{Bhaskar DasGupta \at
           Department of Computer Science, University of Illinois at Chicago, Chicago, IL 60607, USA \\
           Tel.: +312-255-1319\\
           Fax: +312-413-0024\\
           \email{bdasgup@uic.edu}           
           \and
           Marek Karpinski \at
           Department of Computer Science, University of Bonn, Bonn 53113, Germany 
           \email{marek@cs.uni-bonn.de} 
           \and
           Nasim Mobasheri \at
           Department of Computer Science, University of Illinois at Chicago, Chicago, IL 60607, USA \\
           \email{nmobas2@uic.edu}
           \and
           Farzaneh Yahyanejad
           Department of Computer Science, University of Illinois at Chicago, Chicago, IL 60607, USA \\
           \email{fyahya2@uic.edu}
}

\date{Received: date / Accepted: date}

\maketitle

\begin{abstract}
$\delta$-hyperbolic graphs, originally conceived by Gromov in $1987$,
occur often in many network applications; for fixed $\delta$, such graphs are \emph{simply called hyperbolic graphs} and 
include non-trivial interesting classes of ``non-expander'' graphs.
The main motivation of this paper is to investigate the effect of the hyperbolicity measure $\delta$ on expansion and cut-size
bounds on graphs (here $\delta$ need \emph{not} be a constant),
and the asymptotic ranges of $\delta$ for which these results may provide \emph{improved} approximation algorithms 
for related combinatorial problems.
To this effect, we provide
\emph{constructive} bounds on node expansions for $\delta$-hyperbolic graphs
as a function of $\delta$, and show that many witnesses (subsets of nodes) for such expansions can 
be computed efficiently even if the witnesses are required to be nested or sufficiently distinct from each other.
To the best of our knowledge, these are the first such constructive bounds proven.
We also show how to find a large family of $s$-$t$ cuts with {\em relatively small} number of cut-edges
when $s$ and $t$ are sufficiently far apart.
We then provide 
algorithmic consequences of these bounds and their related proof techniques for two problems 
for $\delta$-hyperbolic graphs (\emph{where $\delta$ is
a function $f$ of the number of nodes}, the exact nature of growth of $f$ being dependent on the particular problem considered). 

\end{abstract}


\section{Introduction}

Useful insights for many complex systems such as the world-wide web,
social networks, metabolic networks, and protein-protein interaction networks
can often be obtained by representing them as \emph{parameterized} networks and 
analyzing them using graph-theoretic tools.
Some standard measures used for such investigations include 
degree based measures (\EG, maximum/minimum/average degree or degree distribution)
connectivity based measures (\EG, clustering coefficient, claw-free property, largest cliques or densest sub-graphs), 
and 
geodesic based measures (\EG,  diameter or betweenness centrality).
It is a standard practice in theoretical computer science to investigate and categorize the computational 
complexities of combinatorial problems in terms of ranges of these parameters. For example:
\begin{enumerate}[label=$\blacktriangleright$]
\item
Bounded-degree graphs are known to admit improved approximation as opposed to their arbitrary-degree counter-parts
for many graph-theoretic problems.
\item
Claw-free graphs are known to admit improved approximation as opposed to general graphs
for graph-theoretic problems such as the maximum independent set problem. 
\end{enumerate}
In this paper we consider a \emph{topological} measure called \emph{Gromov-hyperbolicity} (or, simply hyperbolicity for short)
for undirected unweighted graphs that has recently received significant attention 
from researchers in both the graph theory and the 
network science community. This hyperbolicity measure $\delta$ was 
originally conceived in a somewhat different group-theoretic context by Gromov~\cite{G87}.
The measure was first defined for \emph{infinite} continuous metric space 
via properties of geodesics~\cite{book}, but was later also adopted for \emph{finite} graphs.
Lately, there have been a surge of theoretical and empirical works measuring and analyzing the 
hyperbolicity of networks, and many \emph{real-world} networks,such as the
following, have been reported (either theoretically or empirically) to be $\delta$-hyperbolic for $\delta=O(1)$: 
\begin{enumerate}[label=$\blacktriangleright$]
\item
``preferential attachment'' scale-free networks with appropriate scaling (normalization)~\cite{JLB07}, 
\item
networks of high power transceivers in a wireless sensor network~\cite{ALJKZ08}, 
\item
communication networks at the IP layer and at other levels~\cite{PKBV10}, and 
\item
an assorted set of biological and social networks~\cite{ADM}.
\end{enumerate}
%
%
Moreover, extreme congestion at a small number of nodes in a large traffic network 
that uses the shortest-path routing was shown in~\cite{JLBB11} to be caused by a small value of $\delta$ of the network.
On the other hand, theoretical investigations have revealed that \emph{expanders}, \emph{vertex-transitive} graphs and
(for certain parameter ranges) classical \emph{Erd\"{o}s-R\'{e}nyi} random graphs are $\delta$-hyperbolic
only for $\delta=\omega(1)$~\cite{a1,a2,xxx,yyy,NST15}.

A major motivation for this paper is a question of the following type\footnote{This is in contrast
to many research works in this area where one studies the properties of $\delta$-hyperbolic graphs assuming 
$\delta$ to be fixed.}:

\begin{quote}
``{\em
What is the effect of the hyperbolicity measure $\delta$ on expansion and cut-size
bounds on graphs (where $\delta$ is a free parameter and \emph{not} a necessarily a constant)}? 
{\em For what asymptotic ranges of values of $\delta$ can these bounds be used to obtain improved approximation algorithms 
for related combinatorial problems}?''
\end{quote}

Since arbitrarily large $\delta$ leads to the class of \emph{all possible} graphs, 
investigations of this type may eventually provide insights or characterizations of hard graph instances for combinatorial
problems via different asymptotic ranges of values of $\delta$.
To this effect, 
in this paper we further investigate the non-expander properties of hyperbolic networks beyond what 
is shown in~\cite{a1,a2}
and provide constructive proofs of {\em witnesses} (subsets of nodes) satisfying certain expansion or cut-size bounds.
We also provide some algorithmic consequences of these bounds and their related proof techniques for two problems 
related to cuts and paths for graphs.
\emph{A more detailed list of our results is deferred until Section~\ref{sec-over} after the basic definitions and notations}.

\subsection{Basic Notations and Assumptions}
\label{sec-not-def}

We use the following notations and terminologies throughout the paper.
We will simply write $\log$ to refer to logarithm base $2$. 
Our basic input is an ordered triple $\langle G,d,\delta\rangle$ denoting the given {\em connected undirected unweighted} graph 
$G=(V,E)$ 
of hyperbolicity $\delta$ 
in which every node has a degree of 
at most $d>2$.
We will always use the variable $m$ and $n$ to denote 
the number of edges and the number of nodes, respectively, of the given input graph. 
\emph{Throughout the paper, we assume that $n$ is always sufficiently large}.
For notational convenience, \emph{we will ignore floors and ceilings of fractional values} in our theorems and proofs, 
\EG, we will simply write 
$\nicefrac{n}{3}$ instead of 
$\left\lfloor\nicefrac{n}{3}\right\rfloor$ 
or 
$\left\lceil\nicefrac{n}{3}\right\rceil$, since 
this will have \emph{no} effect on the \emph{asymptotic} nature of the bounds.
\emph{We will also make no serious effort to optimize the constants that appear in the bounds in our theorems and proofs}.
In addition, the following notations will be used throughout the paper:
\begin{enumerate}[label=$\blacktriangleright$]
\item
$|\cP|$ is the {\em length} (number of edges) of a path $\cP$ of a graph.
\item
$\overline{u,v}$ is a \emph{shortest path} between nodes $u$ and $v$. In our proofs, any shortest path
can be selected but, once selected, the {\em same} shortest path {\em must} be used in the remaining part of the analysis.
\item
$\dist_H(u,v)$ is the distance (number of edges in a shortest path) 
between nodes $u$ and $v$ in a graph $H$ (and is $\infty$ if there is no path between $u$ and $v$ in $H$). 
\item
$D(H)=\max\limits_{u,v\in V'}\left\{ \dist_H(u,v) \right\}$ is 
the \emph{diameter} of the graph $H=(V',E')$.
Thus, in particular, for our input 
$\langle G,d,\delta\rangle$
there exists two nodes $p$ and $q$ such that 
$\dist_G(p,q)=D(G)\geq\log_d n$. 
\item
For a subset $S$ of nodes of the graph $H=(V',E')$, the {\em boundary} $\partial_H(S)$ of $S$ is the set of nodes 
in $V' \setminus S$ that are connected to {\em at least} one node in $S$, \IE, 
\[
\partial_H(S)=\left\{ u\in V'\setminus S \, | \, v\in S \, \& \, \{u,v\}\in E'\right\}
\]
Similarly, 
for any subset $S$ of nodes, $\cut_H(S)$ denotes the set of edges of $H$ that have \emph{exactly} one end-point in $S$.

The readers should note that our definition of 
$\partial_H(S)$ 
involved the set of the nodes, and \textbf{not the set of edges}, 
that are connected to $S$.
\item
$\cB_H(u,r)$ 
is the set of nodes contained 
in a {\em ball} of radius $r$ centered at node $u$ 
in a graph $H$, \IE, 
$\cB_H(u,r)=\left\{ v \,|\, \dist_H(u,v)\leq r \right\}$ 
\end{enumerate}

\subsection{Formal Definitions of Gromov-hyperbolicity}

Commonly the hyperbolicity measure is defined via geodesic triangles in the following manner.

\begin{definition}[$\delta$-hyperbolic graphs via geodesic triangles]
\label{def-hyperbolic-1}
A graph $G$ has a (Gromov) hyperbolicity of $\delta=\delta(G)$, or simply is $\delta$-hyperbolic, if and only if for every three ordered 
triple of shortest paths $(\overline{u,v},\overline{u,w},\overline{v,w})$, $\overline{u,v}$ lies in a $\delta$-neighborhood
of $\overline{u,w}\,\cup\,\overline{v,w}$, \emph{\IE}, for every node $x$ on $\overline{u,v}$, there exists a node $y$ on 
$\overline{u,w}$ or $\overline{v,w}$ such that $\dist_G(x,y)\leq\delta$.
A $\delta$-hyperbolic graph is simply called a hyperbolic graph if $\delta$ is a constant.
\end{definition}

\begin{definition}[the class of hyperbolic graphs]
Let $\mathcal{G}$ be an infinite collection of graphs. 
Then, $\mathcal{G}$ belongs to the class of hyperbolic graphs if and only if there is an absolute constant $\delta\geq 0$ such
that any graph $G\in\mathcal{G}$ is $\delta$-hyperbolic.
If $\mathcal{G}$ is a class of hyperbolic graphs then any graph $G\in\mathcal{G}$ is simply referred to as a hyperbolic graph.
\end{definition}

There is another alternate but {\em equivalent} (``up to a constant multiplicative factor'') 
way of defining $\delta$-hyperbolic graphs via the following $4$-node conditions.

\begin{definition}[equivalent definition of $\delta$-hyperbolic graphs via $4$-node conditions]
\label{def-hyperbolic-2}
For a set of four nodes $u_1,u_2,u_3,u_4$, let 
$\pi=\left(\pi_1,\pi_2,\pi_3,\pi_4\right)$ be a permutation of $\{1,2,3,4\}$ denoting 
a rearrangement of the indices of nodes such that 
\begin{multline*}
S_{u_1,u_2,u_3,u_4}=\dist_{u_{\pi_1},u_{\pi_2}}+\dist_{u_{\pi_3},u_{\pi_4}} 
\\
\leq M_{u_1,u_2,u_3,u_4}=\dist_{u_{\pi_1},u_{\pi_3}}+\dist_{u_{\pi_2},u_{\pi_4}} 
\\
\leq L_{u_1,u_2,u_3,u_4}=\dist_{u_{\pi_1},u_{\pi_4}}+\dist_{u_{\pi_2},u_{\pi_3}}
\end{multline*}
and let 
$\rho_{u_1,u_2,u_3,u_4} = \dfrac{L_{u_1,u_2,u_3,u_4}-M_{u_1,u_2,u_3,u_4}}{2}$.
Then, $G$ is $\delta$-hyperbolic if and only if
\[
\delta = \delta(G) = \max\limits_{u_1,u_2,u_3,u_4\in V} \big\{ \,\rho_{u_1,u_2,u_3,u_4}\big\}.
\]
\end{definition}

It is well-known (\EG, see~\cite{book}) that 
Definition~\ref{def-hyperbolic-1} and Definition~\ref{def-hyperbolic-2}
of $\delta$-hyperbolicity are equivalent in the sense that they are related
by a constant multiplicative factor, \IE, there is an absolute constant $c>0$ such that if a graph $G$ is $\delta_1$-hyperbolic and 
$\delta_2$-hyperbolic via Definition~\ref{def-hyperbolic-1} and Definition~\ref{def-hyperbolic-2}, respectively, then 
$\delta_1/c \leq \delta_2 \leq c \,\delta_1$.
Since constant factors are not optimized in our proofs, we will use {\em either} of the two definitions of
hyperbolicity in the sequel as deemed more convenient.
Using Definition~\ref{def-hyperbolic-2} and casting the resulting computation as a $(\max,\min)$ matrix multiplication 
problem allows one to compute $\delta(G)$ and a $2$-approximation of $\delta(G)$ in 
$O\left(n^{3.69}\right)$ and in $O\left(n^{2.69}\right)$ time, respectively~\cite{ipl15}.
Several routing-related problems or the diameter estimation problem become easier if
the network is hyperbolic~\cite{CE07,CDEHV08,CDEHVX12,GL05}.
For a discussion of
properly scaled $4$-node conditions that yield a variety of (non necessarily hyperbolic) geometries, see~\cite{JLA11}.

\subsubsection{Remarks on Topological Characteristics of Hyperbolicity Measure $\delta$}
\label{sec-topo}

Even though the hyperbolicity measure $\delta(G)$ is often referred to as a ``tree-like'' measure, 
$\delta(G)$ enjoys many non-trivial topological characteristics. 
For example:
\begin{description}
\item[$\star$]
{\bf The ``$\delta(G)=o(n)$'' property is not hereditary (and thus also not monotone)}.
For example, see \FI{fig-nontriv}, which 
also shows that 
removing a single node 
or edge can increase/decrease the value of $\delta$ \emph{very sharply}. 
\smallskip
\item[$\star$]
{\bf ``Close to hyperbolic topology'' is not necessarily the same as ``close to tree topology''}.
For example,
\emph{all} bounded-diameter graphs have $\delta=O(1)$ irrespective of whether they are tree or not
(however, graphs with $\delta=O(1)$ need \emph{not} be of bounded diameter). In general, even for small $\delta$, the metric
induced by a $\delta$-hyperbolic graph may be quite far from a tree metric~\cite{CDEHV08}.
\smallskip
\item[$\star$]
{\bf Hyperbolicity is not necessarily the same as tree-width}.
A similar popular measure used in both the bioinformatics and theoretical computer science literature 
is the treewidth measure first introduced by Robertson and Seymour~\cite{RS83}.
Many $\NP$-hard problems on general networks in fact allow polynomial-time solutions if restricted to classes of networks 
with bounded treewidth~\cite{B88}.
However, as observed in~\cite{MSV11} and elsewhere, the two measures are quite different in nature and \emph{not} correlated.
\end{description}

\begin{figure}[htbp]
\centerline{\includegraphics{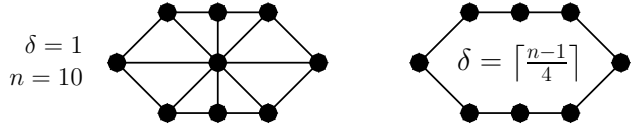}}
\caption{\label{fig-nontriv}The ``$\delta(G)=o(n)$'' property is not hereditary.}
\end{figure}

\noindent
Examples of hyperbolic graph classes 
(\IE, when $\delta$ is a constant)
include {\em trees}, {\em chordal graphs}, {\em cactus of cliques}, 
{\em AT-free} graphs, {\em link graphs of simple polygons}, 
and {\em any} class of graphs with a {\em fixed} diameter, whereas examples of non-hyperbolic graph classes
(\IE, when $\delta$ is not a constant)
include {\em expanders}, {\em simple cycles}, and, 
for some parameter ranges, the {\em Erd\"{o}s-R\'{e}nyi random graphs}. 

Note that if $G$ is $\delta$-hyperbolic then $G$ is also $\delta'$-hyperbolic for any $\delta'>\delta$ (cf.\ 
Definition~\ref{def-hyperbolic-1}).
In this paper, to avoid division by zero in terms involving $\nicefrac{1}{\delta}$, we will assume $\delta>0$.
In other words, we will treat a $0$-hyperbolic graph (a tree) as a $\frac{1}{2}$-hyperbolic graph in the analysis.

\subsection{Relevant Known Results for Gromov Hyperbolicity}

We summarize relevant known results that are used in this paper below; 
many of these results appear in several prior works, \EG,~\cite{a1,a2,book,G87,ADM}.
\FI{fig-kill-short} pictorially illustrates these results.

\begin{fact}[Cylinder removal around a geodesic]{\rm\cite{a2}}
\label{fact-cylinder}
Assume that $G$ is a $\delta$-hyperbolic graph.
Let $p$ and $q$ be two nodes of $G$ such that $\dist_G(p,q)=\beta>6$, and 
let $p',q'$ be nodes 
on a shortest path between $p$ and $q$ such that $\dist_G(p,p')=\dist_G(p',q')=\dist_G(q',q)={\beta}/{3}$. 
For any $0<\alpha<\nicefrac{1}{4}$, let $\cC$ be set of nodes at a distance of $\alpha \beta-1$
of a shortest path $\overline{p',q'}$ between $p'$ and $q'$, \IE, 
let 
$
\cC = \left\{ u \,|\, \exists \, v \in \overline{p',q'} \colon \dist_G(u,v)=\alpha \beta-1\right\} 
$.
Let $G_{-\cC}$ be the graph obtained from $G$ by removing the nodes in $\cC$.
Then, 
$
\dist_{G_{-\cC}}(p,q) \geq 
({\beta}/{60})\,2^{\,\alpha \beta/\delta}
$.
\end{fact}

\begin{fact}[Exponential divergence of geodesic rays]\label{fact-expo}
{\em\textbf{[Simplified r\-e\-f\-o\-r\-m\-u\-l\-a\-t\-i\-o\-n of \cite[Theorem~10]{ADM}]}}
Assume that $G$ is a $\delta$-hyperbolic graph.
Suppose that we are given the following:
\begin{itemize}[itemsep=0.1ex]
\item
three integers $\kappa\geq 4$, $\alpha>0$, $r>3\kappa\delta$, and 
\smallskip
\item
five nodes $v,u_1,u_2,u_3,u_4$ such that
$\dist_G(v,u_1)=\dist_G(v,u_2)=r$, 
\\
$\dist_G(u_1,u_2)\geq 3 \kappa\delta$,
$\dist_G(v,u_3) \! = \! \dist_G(v,u_4) \! = \! r+\alpha$, and 
$\dist_G(u_1,u_4) \! = \! \dist_G(u_2,u_3)=\alpha$.
\end{itemize}
Consider any path $\cQ$ between $u_3$ and $u_4$ that does not involve a node in 
$\bigcup_{\,0\,\leq\, j\,\leq\, r+\alpha}\hspace*{-0.0in}\cB_{G}(v,j)$. Then, 
the length $|\cQ|$ of the path $\cQ$ satisfies 
$\displaystyle | \cQ | \! > \! 2^{ \frac{\alpha} { 6\,\delta } + \kappa + 1 }$.
\end{fact}

\begin{figure}[htbp]
\includegraphics{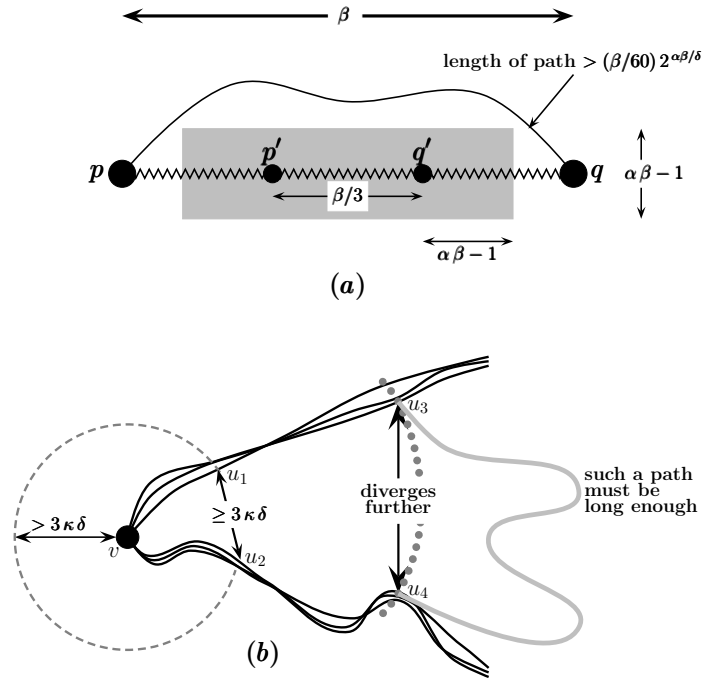}
\caption{\label{fig-kill-short}
($a$) 
Illustration of Fact~\ref{fact-cylinder}.
By growing the shaded region and removing nodes in its boundary, one can selectively extract 
longer paths in the graph (\IE, the length of a shortest path between $p$ and $q$ increases when the nodes in the boundary 
of the shaded region are removed and the increase of the length of such a shortest path is more the larger the shaded region is). 
Translating the region slightly does not change this property much.
($b$) 
Illustration of Fact~\ref{fact-expo}.
Geodesic rays diverging sufficiently cannot connect back without using a sufficiently long path.}
\end{figure}

\section{Overview of Our Results}
\label{sec-over}

Before proceeding with formal theorems and proofs, we first provide an informal non-technical intuitive overview 
of our results.
\begin{enumerate}[label=$\blacktriangleright$]
\item
Our first two results in Section~\ref{n-e-sec} provide upper bounds for {\em node expansions} for the triple
$\langle G,d,\delta\rangle$ as a function of $n$, $d$, and $\delta$. These two results, namely Theorem~\ref{main-nested}
and Theorem~\ref{main3}, 
provide absolute bounds and show that many {\em witnesses} (subset of nodes) satisfying such expansion bounds 
can be found efficiently in polynomial time satisfying two additional criteria: 
\smallskip
\begin{enumerate}[label=$\triangleright$]
\item
the witnesses (subsets) form a nested family, or
\smallskip
\item
the witnesses have {\em limited} overlap in the sense that 
every subset has a certain number of ``private'' nodes \emph{not} contained in any other subset.
\end{enumerate}
\smallskip
These bounds also imply in an obvious manner corresponding upper bounds for the {\em edge-expansion} of $G$ and 
for the smallest non-zero eigenvalue of the Laplacian of $G$.

\smallskip

In Remark~\ref{rem1}, we provide an explanation of the asymptotics of these bounds in comparison to
expander-type graphs. 
For example, if $\delta$ is fixed (\IE, $G$ is hyperbolic) then 
$d$ has to be increased to {\em at least} $2^{\Omega\left(\sqrt{  {\log\log n}\,/\,{\log\log\log n}  }\,\right)}$
to get a positive non-zero Cheeger constant, whereas if $d$ is fixed then 
$\delta$ need to be at least $\Omega\left(\log n\right)$
to get a positive non-zero Cheeger constant (this last implication also follows from the results in~\cite{a1,a2}).

\smallskip

\item
Our last result in Section~\ref{sec-disjoint}, namely Lemma~\ref{const-thm}, deals with the {\em absolute} size 
of $s$-$t$ cuts in hyperbolic graphs, and shows that 
a large family of $s$-$t$ 
cuts having at most $d^{\,O(\delta)}$ cut-edges 
can be found in polynomial time 
in $\delta$-hyperbolic graphs when 
$d$ is the maximum degree of any node except $s$, $t$ and any node within a distance of 
$35\,\delta$ of $s$ and 
the distance between $s$ and $t$ is at least $\Omega(\delta\log n)$.
This result is later used in designing the approximation algorithm for 
minimizing bottleneck edges
in Section~\ref{sec-size}.

\smallskip

\item
In Section~\ref{sec-appl} we discuss some applications of these bounds in designing improved approximation algorithms 
for two graph-theoretic problems
for $\delta$-hyperbolic graphs when $\delta$ does not grow too fast as a function of $n$: 

\smallskip

\begin{enumerate}[label=$\triangleright$]
\item
We show in Section~\ref{sec-bott} (Lemma~\ref{equi2}) that the 
problem of identifying vulnerable edges in network designs by {\em minimizing shared edges} admits an improved 
approximation provided $\delta=o(\log n/\log d)$. We do so by relating it to a hitting set problem for  
{\em size-constrained cuts} (Lemma~\ref{equi}) and providing an improved approximation for this 
latter problem (Lemma~\ref{main4}). We also observe
that obvious greedy strategies fail for such problems miserably.

\smallskip

\item
In Section~\ref{sec-small-expansion}
we provide a polynomial-time solution (Lemma~\ref{thm-sse-hyper}) 
for a type of small-set expansion problem
originally proposed by Arora, Barak and Steurer~\cite{ABS10}
for the case when $\delta$ is sub-logarithmic in $n$.
\end{enumerate}

\smallskip

\item
Finally, in Section~\ref{sec-concl} we conclude with some interesting future research questions.
\end{enumerate}

\section{Effect of $\delta$ on Expansions and Cuts in $\delta$-hyperbolic Graphs}
\label{n-e-sec}

The two results in this section are related to the node (or edge) expansion ratios of a graph that is $\delta$-hyperbolic
for some (not necessarily constant) $\delta$.
The following definitions are standard in the graph theory literature and repeated here only for the sake of completeness.

\begin{definition}[Node and edge expansion ratios of a graph]~\\
\noindent
{\bf (a)}
The \emph{node expansion ratio} 
$h_G(S)$ of a subset $S$ of at most $\nicefrac{\left|V\right|}{2}$ nodes of a graph $G=(V,E)$ is defined as 
$h_G(S)=\frac{|\,\partial_G(S)\,|}{|\,S\,|}$.
If $h_G(S)>c$ for some constant $c>0$ and for all subsets $S$ of at most $\nicefrac{\left|V\right|}{2}$ nodes
then we call $G$ a node-expander.

\smallskip
\noindent
{\bf (b)}
The edge expansion ratio $\g_H(S)$ of a subset $S$ of at most $\nicefrac{\left|V\right|}{2}$ nodes of a graph $G=(V,E)$ is defined as 
$\g_G(S)=\frac{|\,\cut_G(S)\,|}{|\,S\,|}$. 
If $h_G(S)>c$ for some constant $c>0$ and for all subsets $S$ of at most $\nicefrac{\left|V\right|}{2}$ nodes
then we call $G$ an edge-expander (or sometimes simply an expander). 
\end{definition}

\begin{definition}[Witness of node or edge expansions]
A witness of a node (respectively, edge) expansion bound of $c$ of a graph $G=(V,E)$ 
is a subset $S$ of at most $\nicefrac{\left|V\right|}{2}$ nodes of $G$
such that $h_G(S)\leq c$ (respectively, $\g_G(S)\leq c$).
\end{definition}

\smallskip

\noindent
{\bf Notation} 
$\h_G= \hspace*{-0.25in} \min\limits_{\,\,\,S\subset V\,\colon\,|S|\leq |V|/2} \Big\{ h_G(S)\Big\}$
will denote the \emph{minimum} node expansion of a graph $G=(V,E)$.

\medskip

For any graph $G=(V,E)$,
any subset $S$ containing exactly $\nicefrac{|V|}{2}$ nodes has $|\,\partial_G(S)\,|\leq \nicefrac{|V|}{2}$, and thus
$0<\h_G=\min_{S\subset V\,\colon\,|S|\leq |V|/2} \left\{ h_G(S)\right\}\leq 1$
\emph{All our expansion bounds in this section will be stated for node expansions only}.
Since $\g_G(S)\leq d\,h_G(S)$ for any graph $G$ whose nodes have a maximum degree of $d$, our bounds
for node expansions translate to some corresponding bounds for the edge expansions as well.

\subsection{Nested Family of Witnesses for Node/Edge Expansion}

An ordered family of sets $S_1,S_2,\dots,S_\ell$ is called \emph{nested} if 
$S_1\subset S_2\subset\dots\subset S_\ell$.
Our goal in this subsection is to find a large nested family of subsets of nodes 
with good node expansion bounds.

For two nodes $p$ and $q$ of a graph $G=(V,E)$, 
a cut $S$ of $G$ that ``separates $p$ from $q$''
is a subset $S$ of nodes containing $p$ but not containing $q$, and the 
set of cut edges $\cut_G(S,p,q)$ corresponding to the cut $S$ 
is the set of edges with exactly one end-point in $S$, \IE, 
\[
\cut_G(S,p,q)=\Big\{ 
\,\big\{u,v\big\} \,\,|\,\,
p,u\in S \,\,\text{and}\,\, q,v\in V\setminus S
\, \Big\}
\]
Recall that $d$ denotes the maximum degree of any node in the given graph $G$.

\begin{theorem}
\label{main-nested}
For any constant $0<\mu<1$, the following result holds for $\langle G,d,\delta\rangle$.
Let $p$ and $q$ be any two nodes of $G$ and let $\Delta=\dist_G(p,q)$.
Then, there exists at least $t=\max\left\{ \frac{\Delta^\mu}{56\log d},\,1 \right\}$ subsets of nodes 
$\emptyset\subset S_1 \subset S_2 \subset \dots \subset S_t\subset V$, 
each of at most ${n}/{2}$ nodes,
with the following properties:
\begin{enumerate}[label=$\blacktriangleright$]
\item
$\displaystyle
\forall\, j\in\{1,2,\dots,t\} \,\colon$

\hfill
$\displaystyle
h_G\left(S_j\right) 
\leq 
\min\left\{ 
{8\ln\left({n}/{2}\right)} / {\Delta}, \,
\max\left\{ \left({1} / {\Delta}\right)^{1-\mu},\,
\dfrac { 500 \ln n }{ \Delta\,2^{\,   \frac{ \Delta^{\mu}} {28\,\delta\,\log(2d) }    } }
\, \right\}
\, \right\}
$.
\smallskip
\item
All the subsets can be found in a total of $O\left(n^3\log n + mn^2\right)$ time.
\smallskip
\item
Either all the subsets $S_1,S_2,\dots,S_t$ contain the node $p$, or 
all of them contain the node $q$.
\end{enumerate}
\end{theorem}

\begin{corollary}
\label{cor1}
Letting $p$ and $q$ be two nodes such that $\dist_G(p,q)=D(G)=D$ realizes the diameter of the graph $G$, we get
the bound: 
\[
h_G\left(S_j\right) 
\leq 
\min\left\{ 
{\textstyle \frac {8\ln\left( {n}/{2} \right)} {D} }, \,
\max\left\{ \left( {\textstyle \frac {1}{D} } \right)^{1-\mu},\,
( { 500 \ln n } ) / \left( { D\,2^{\,   \frac { D^{\mu}} {28\,\delta\,\log(2d) }   } } \right)
\, \right\}
\, \right\}
\]
Since $D>{\log n} / {\log d}$, the above bound implies:
\begin{gather}
\h_G<
\max\left\{ \left( {\log d} / {\log n}\right)^{1-\mu},\,
( { 500 \log d } ) / \left( { 2^{\,     {\log^\mu n } / \left( { 28\,\delta \, {\log^{1+\mu} (2d)} } \right)\, } } \right) \, \right\}
\label{bound-expl}
\end{gather}
\end{corollary}

\begin{remark}
\label{rem1}
The following observations may help the reader to understand the asymptotic nature of the bound in~\eqref{bound-expl}.

\smallskip
\noindent
{\bf (a)}
The first component of the bound is $O\left({1}/{\log^{1-\mu} n}\right)$ for fixed $d$,
and is $\Omega(1)$ only when $d=\Omega(n)$.

\smallskip
\noindent
{\bf (b)}
To better understand the second component of the bound, consider the following cases (recall that $\h_G=\Omega(1)$ for an expander):
\begin{enumerate}[label=$\blacktriangleright$]
\item
Suppose that the given graph is a hyperbolic graph of constant maximum degree, \IE, both $\delta$ and $d$ are constants.
In that case, 
\[
( { 500 \log d } ) / \left( { 2^{\,  \frac {\log^\mu n } { 28\,\delta \, {\log^{1+\mu} (2d) } } \,    } } \right)
=
O \left(
{ 1 } / \left( { 2^{ \, O(1) \, \log^\mu n } } \right) \,
\right)
=
O \left(
{ 1 } / { \mathrm{polylog} (n) }\,
\right)
\]
\item
Suppose that the given graph is hyperbolic but the maximum degree $d$ is arbitrary. 
In that case, 
\begin{multline*}
( { 500 \log d } ) / \left( { 2^{\,  \frac{\log^\mu n } { 28\,\delta \, {\log^{1+\mu} (2d) } }   } } \right)
=
O \left(
{ \log d } / \left( { 2^{\,  O(1) \, {\log^\mu n } / { \log^{1+\mu} d  }   } } \right) \,
\right)
\\
=
O\left(
{ \log d }
/
{
\mathrm{polylog}(n)^{{1} / {\log^{1+\mu} d}}
}
\right)
\end{multline*}
and thus $d$ has to be increased to at least $2^{\Omega\left(\,\sqrt{  {\log\log n} / {\log\log\log n}  }\,\right)}$ to get
a constant upper bound.
\smallskip
\item
Suppose that the given graph has a constant maximum degree but not necessarily hyperbolic ({\em\IE}, $\delta$ is
arbitrary).
In that case, 
\[
( { 500 \log d } ) / \left( { 2^{\, ( { \log^\mu n } ) / \left( { 28\,\delta \, {\log^{1+\mu} (2d) }} \right) \,} } \right)
= 
O \left(
{1} / {
2^{ \, O(1) {\log^\mu n } / {\delta} }
}
\right)
\]
and thus $\delta$ need to be at least $\Omega\left(\log^\mu n\right)$ to get
a constant upper bound.
\end{enumerate}
\end{remark}


\subsubsection{Proof of Theorem~\ref{main-nested}}

Proof of the main bounds in Theorem~\ref{main-nested} uses the same cylinder or ball removing techniques
as used in~\cite{a1,a2} in showing that hyperbolic graphs are not expanders.
However, several technical complications arise when we try to find these witnesses while optimizing 
the corresponding expansion bounds. 
The time-complexity of finding our witnesses are discussed at the very end of our proof.

\bigskip
\noindent
(I)
{\bf Proof of the easy part of the bound, \IE, $h_G\left(S_j\right) \leq ({8\ln\left({n}/{2}\right)})/{\Delta}$}

\medskip
This proof is straightforward and \emph{provided for the sake of completeness}.
Assume that $\Delta>\left( 8 \ln \left({n}/{2} \right) \right)^{1/\mu}$ since otherwise there is no need to prove this bound.
Assume, without loss of generality, that 
\\
$\left| \cB_G\big(p,{\Delta}/{2}\big) \right| \leq 
\min\left\{ \, \left| \cB_G\big(p,{\Delta}/{2}\big) \right|,\,
\left| \cB_G\big(q,{\Delta}/{2}\big) \right|\,
\right\}\leq {n}/{2}$.
Consider the sequence of balls 
$\cB_G(p,r)$ for $r=0,1,2\dots,{\Delta}/{2}$. Then it follows that 
\begin{multline*}
{n}/{2} > \left| \,\cB_G\left(p,{\Delta}/{2}\right)\, \right| \geq 
\hspace*{-0.2in}
\prod_{\ell=0}^{ \hspace*{0.2in} ({\Delta}/{2})-1} 
\hspace*{-0.2in}
\big(  1+h_G \left( \, \cB_G\left(p,\ell\right) \, \right)  \big)
\\
\geq
\hspace*{-0.2in}
\prod_{\ell=0}^{ \hspace*{0.2in} ({\Delta}/{2})-1} 
\hspace*{-0.2in}
\bee^{h_G \left( \, \cB_G\left(p,\ell\right) \, \right) /2}
=
\bee^{
     \hspace*{-0.2in}
     \sum\limits_{\ell=0}^{ \hspace*{0.2in} ({\Delta}/{2})-1} 
     \hspace*{-0.2in}
		 h_G \left( \, \cB_G\left(p,\ell\right) \, \right) /2 }
\\
\Rightarrow \,
\ln \left( {n}/{2} \right) 
>
\hspace*{-0.25in}
\sum\limits_{\ell=0}^{ \hspace*{0.2in} ({\Delta}/{2})-1} 
\hspace*{-0.25in}
h_G \left( \, \cB_G\left(p,\ell\right) \, \right) /2
\,\Rightarrow \,
\frac { \sum_{\ell=0}^{ ({\Delta}/{2})-1} \! h_G \left( \, \cB_G\left(p,\ell\right) \, \right) } {\Delta/2}
<
 \frac {4\ln\left({n}/{2}\right)} {\Delta}
\end{multline*}
By a simple averaging argument, there must now exist ${\Delta}/{4}>
\max\left\{ \frac {\Delta^\mu} {56\log d},\,1 \right\}$ distinct balls (subsets of nodes) 
$
\cB_G\left(p,r_1\right)\subset
\cB_G\left(p,r_2\right)\subset
\dots\subset
\cB_G\left(p,r_{\nicefrac{\Delta}{4}}\right)
$ 
such that 
$
\left| \, \cB_G\left(p,r_j\right) \, \right|
<
({8\ln\left({n}/{2}\right)})/{\Delta}
$
for $j=1,2,\dots,{\Delta}/{4}$. It is straightforward to see that 
these balls can be found 
within the desired time complexity bound.

\bigskip
\noindent
(II)
{\bf Proof of the difficult part of the bound,}

{\bf \IE, $h_G\left(S_j\right) \leq 
\max\left\{ \left({1}/{\Delta}\right)^{1-\mu},\,
\dfrac { 500 \ln n }{ \Delta\,2^{\,   \frac{ \Delta^{\mu}} {28\,\delta\,\log(2d) }    } }
\, \right\}$}

\medskip
\noindent
(II-a)
{\bf The easy case of $\Delta=O(1)$}

\medskip
If $\Delta=c$ for any some constant $c\geq 1$ (independent of $n$) 
then, since $\delta\geq {1}/{2}$, $d>1$ and $n$ is sufficiently large, we have 
$
( { 500 \ln n }) / \left( { \Delta\,2^{\,  {\Delta^{\mu}} / ( {28\,\delta\,\log(2d) } ) } } \right)
>
( { 500 \ln n } ) / \left( { \Delta\,2^{\,({1}/{14}) \Delta^{\mu}} } \right)
>  1 
$.
Thus, any subset of ${n}/{2}$ nodes containing $p$ satisfies the claimed bound, and 
the number of such subsets is 
$\displaystyle\binom{\frac{n}{2}-1}{n-2}\gg t$.

\medskip
\noindent
(II-b)
{\bf The case of $\Delta=\omega(1)$}

\medskip
Otherwise, assume that $D(n)=\omega(1)$, \IE, $\lim_{n\to\infty}D(n)>c$ for any constant $c$. 
Let $p',q'$ be nodes 
on a shortest path between $p$ and $q$ such that $\dist_G(p,p')=\dist_G(p',q')=\dist_G(q',q)={\Delta}/{3}$. 
The following initial value of the parameter $\alpha$ is crucial to our analysis\footnote{We will later need to
vary the value of $\alpha$ in our analysis.}: 
\begin{equation}
\alpha
= 
\alpha_0
=
{1} / \left( {7\,{\Delta}^{1-\mu}\,\log (2 d)} \right)
\label{alpha-value}
\end{equation}
Note that $0<\alpha_0<{1}/{4}$. 
Let $\cC$ be set of nodes at a distance of $\lfloor \alpha \Delta\rfloor>\alpha \Delta - 1$
of a shortest path $\overline{p',q'}$ between $p'$ and $q'$. Thus,
\begin{gather}
\cC = \left\{ u \,|\, \exists \, v \in \overline{p',q'} \colon \dist_G(u,v)=\lceil \alpha \Delta\rceil \, \right\} 
\,\Rightarrow\,
\left| \,\cC \, \right| 
\leq
({\Delta}/{3})\,d^{\,\lfloor\alpha \Delta\rfloor}
< ({\Delta}/{3})\,d^{\,\alpha \Delta}
\label{eq:a1}
\end{gather}
Let $G_{-\cC}$ be the graph obtained from $G$ by removing the nodes in $\cC$.
Fact~\ref{fact-cylinder} implies:
\begin{gather}
\dist_{G_{-\cC}}(p,q) \geq 
({\Delta}/{60})\,2^{\,\alpha \Delta/\delta}
\label{eq:a2}
\end{gather}
Let $\cB_G(p,r)$ be the ball of radius $r$ centered at node $p$ in $G$ with
$\left| \,\cB_G(p,r)\, \right|\leq {n}/{2}$, and let
$
\have(p,j)
\,\,\stackrel{\mathrm{def}}{=}\,\,
{ \left( \sum_{\ell=0}^{j-1} \cB_{G}(p,\ell) \right) } / { j }
$.
Then, since 
$\left| \,\cB_G(p,0)\, \right|=1$ and
$\frac{\left| \cB_G(p,r)\right|}{\left| \cB_G(p,r-1)\right|} = 1+h_G\left(\cB_G(p,r-1)\right)$, we have 
\begin{multline}
\left| \,\cB_G(p,r)\, \right| =
\prod\limits_{j=0}^{r-1} \left( 1+h_G\left( \cB_G(p,j) \right) \right)
\\
\geq
\prod\limits_{j=0}^{r-1} \bee^{h_G\left( \cB_G(p,j) \right) /2}
=
\bee^{\sum\limits_{j=0}^{r-1}h_G \left( \cB_G(p,j) \right)/2}
=
\bee^{r\,\have(p,r)/2}
\label{eq-new5}
\end{multline}
Assume without loss of generality that\footnote{Note that if there is no path between nodes $p$ and $q$ in $G_{-\cC}$ then 
$\dist_{G_{-\cC}}(p,q)=\infty$ and
$\cB_{G_{-\cC}}\left(p,{\dist_{G_{-\cC}}(p,q)}/{2}\right)$ and 
$\cB_{G_{-\cC}}\left(q,{\dist_{G_{-\cC}}(p,q)}/{2}\right)$ contain all the nodes reachable from $p$ and $q$, respectively, 
in $G_{-\cC}$.}
\begin{gather}
\left| \, \cB_{G_{-\cC}}\left(p,{\dist_{G_{-\cC}}(p,q)}/{2}\right) \, \right|
\leq
\left| \, \cB_{G_{-\cC}}\left(q,{\dist_{G_{-\cC}}(p,q)}/{2}\right) \, \right|
\leq
\left( {n-|\cC|}\right) / {2}
<
\frac {n} {2}
\label{eq:less}
\end{gather}

\medskip
\noindent
{\bf Case 1}:
There is 
a set of $t$ distinct indices 
{
\def\OldComma{,}
    \catcode`\,=13
    \def,{%
      \ifmmode%
        \OldComma\discretionary{}{}{}%
      \else%
        \OldComma%
      \fi%
    }%
$\left\{ i_1,i_2,\dots,i_{t} \right\} \subseteq 
\{0,1,\dots,\dist_{G_{-\cC}}(p,q)/2 \}$
}
such that, 
$i_1<i_2<\dots<i_{t}$
and, for all $1\leq s\leq t$, 
$
h_G\left( \, \cB_G(p,i_s)\, \right) 
=
h_G\left( \, \cB_{G_{-\cC}}(p,i_s)\, \right) 
\leq \left({1}/{\Delta}\right)^{1-\mu}
$ 
(see \FI{p1-fig}$\!$(\emph{a})). 
Then, the subsets 
$\cB_G(p,i_1)\subset \cB_G(p,i_2)\subset\dots\subset \cB_G(p,i_t)$
satisfy our claim.

\begin{figure*}[htbp]
\centerline{\includegraphics{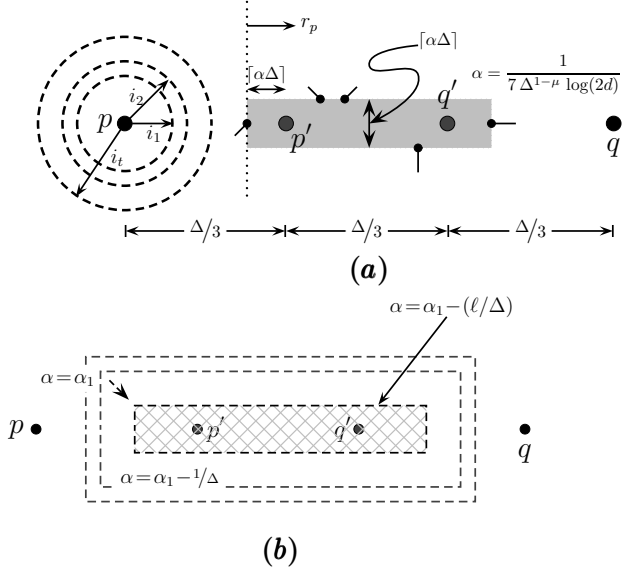}}
\caption{\label{p1-fig}Illustration of various cases in the proof of Theorem~\ref{main-nested}.
($a$) Case $1$. 
Nodes on the boundary of the lightly shaded region belong to 
$\cC_{\alpha_1 \Delta}$.
($b$) 
Case $2$.
Nodes on the boundary of the lightly cross-hatched region belong to 
$\cC_{\alpha_1 \Delta-\ell}$.
}
\end{figure*}

\medskip
\noindent
{\bf Case 2}: Case 1 does not hold. 
In this case, we have 
\begin{multline}
\hspace*{-0.5in}
\sum\limits_{\ell=0}^{ \hspace*{0.4in} ({\Delta}/{3})-\alpha \Delta-1}
\hspace*{-0.45in}
h_G\left( \cB_G(p,\ell) \right)
>
\left(
   \left( {\dist_{G_{-\cC}}(p,q)} / {2} \right) - (t-1)
\right)
\,
\left({1}/{\Delta}\right)^{1-\mu}
\\
>
\left(
    ({\Delta}/{3})-\alpha \Delta - t
\right)
\left({1}/{\Delta}\right)^{1-\mu}
>
{\Delta^\mu}/{4}
\label{smallball}
\end{multline}
Let $r_p$ be the {\em least} integer such that 
$
\cB_{G_{-\cC}}\left(p,r_p\right)
=
\cB_{G_{-\cC}}\left(p,r_p+1\right)
$.
Since $G$ is a connected graph and,  
for all $r\leq ({\Delta}/{3})-\alpha \Delta$ we have 
$\cB_G(p,r)\cap\cC=\emptyset\equiv \cB_{G_{-\cC}}(p,r)=\cB_G(p,r)$
we have $r_p\geq ({\Delta}/{3})-\alpha \Delta$ (see \FI{p1-fig}$\!$(\emph{a})).

\medskip
\noindent
{\bf Failure of the current strategy}

\smallskip
Note that it is possible that 
$r_p$ is precisely $({\Delta}/{3})-\alpha \Delta$ or 
not too much above it (this could happen when $p$ is disconnected from $q$ in $G_{-\cC}$).
\emph{Consequently, we may not be able to use our current technique of enlarging the ball
$\cB_{G_{-\cC}}\left(p,r\right)$ 
for $r$ beyond 
$({\Delta}/{3})-\alpha \Delta$ to get the required number of subsets of nodes as claimed in the theorem}.
A further complication arises because, for $r>({\Delta}/{3})-\alpha \Delta$,
expansion of the balls 
$\cB_{G_{-\cC}}\left(p,r\right)$ 
in $G_{-\cC}$ may differ from that in $G$, \IE, 
$h_G \left( \,\cB_{G_{-\cC}}\left(p,r\right)  \, \right)$ \emph{need not be the same as} 
$h_{G_{-\cC}} \left( \,\cB_{G_{-\cC}}\left(p,r\right)  \, \right)$.

\medskip
\noindent
{\bf Rectifying the current strategy}

\smallskip
We now change our strategy in the following manner.
Let us write $r_p$ as $r_{p,\alpha \Delta}$ to show its dependence on $\alpha \Delta$ and let
$\alpha_1=\frac{1}{14\,{\Delta}^{1-\mu}\,\log (2 d)}$.
Vary $\alpha$ from $\alpha=\alpha_1$ to $\alpha=\alpha_1/2$ in steps of $-1/\Delta$, and 
consider the sequence of values 
$r_{p,\,\alpha_1 \! \Delta},r_{p,\,\alpha_1\!  \Delta-1},\dots, 
r_{p,\,\nicefrac{\alpha_1 \Delta}{2}}$.
Let $\cC_{\alpha_1 \Delta-\ell}$ denote the set of nodes in $\cC$ when $\alpha$ is set equal to $\alpha_1-(\ell/\Delta)$
for $\ell=0,1,2,\dots,\alpha_1\Delta/2$ (see \FI{p1-fig}$\!$(\emph{b})). 
Consider the two sets of nodes $\cC_{\alpha_1 \Delta-\ell}$ and $\cC_{\alpha_1 \Delta-\ell'}$ with $\ell<\ell'$.
Obviously, $\cC_{\alpha_1 \Delta-\ell}\neq\cC_{\alpha_1 \Delta-\ell'}$ for any $\ell\neq \ell'$.

\medskip
\noindent
{\bf Case 2.1 (relatively easier case)}:
Removal of each of the set of nodes 
$\cC_{\alpha_1 \Delta}, \cC_{\alpha_1 \Delta-1}, \dots, \cC_{ \left({\alpha_1 \Delta}\right)/{2}}$
disconnects $p$ from $q$ in the corresponding graphs 
$G_{-\cC_{\alpha_1 \Delta} }, G_{-\cC_{\alpha_1 \Delta-1} }, \dots, G_{-\cC_{ \left({\alpha_1 \Delta}\right)/{2}} }$,
respectively.

\smallskip
Then, for any $0\leq\ell\leq \left({\alpha_1 \Delta}\right)/{2}$, we have 
\begin{gather*}
r_{p,\alpha_1 \Delta-\ell}\geq ({\Delta}/{3})-\alpha_1 \Delta+\ell \geq ({\Delta}/{3})-\alpha_1 \Delta
\\
\left| \cB_{G_{-\cC_{\alpha_1 \Delta -\ell} }} \! \left(p,r_{p,\alpha_1 \Delta -\ell} \right) \right|
\!
>
\!
\underset{ \hspace*{0.8in} \text{\small by \eqref{eq-new5} } }
{
\left| \cB_{G_{-\cC_{\alpha_1 \Delta } }} \!\! \left(\!p, \frac{\Delta}{3}-\alpha_1 \Delta\!\right) \right|
\!
\geq
\!
\bee^{\frac{1}{2}\hspace*{-0.55in}\sum\limits_{j=0}^{\hspace*{0.5in} ({\Delta}/{3})-\alpha_1 D -1} \hspace*{-0.5in}h_G \left( \cB_G(p,j) \right)}
}
\underset{ \hspace*{-0.3in} \text{\small by \eqref{smallball} } }
{
>
\bee^{ \nicefrac{\Delta^\mu}{8} }
}
\\
\left| \,\partial_G \left( \cB_{G_{-\cC_{\alpha_1 \Delta -\ell} }}\left(p,r_{p,\alpha_1 \Delta -\ell}\right)  \right) \, \right|
\leq
\left|\,\cC_{\alpha_1 \Delta -\ell}\,\right|
\leq
\underset{\hspace*{-0.2in} \text{\small by \eqref{eq:a1}} }
{
|\,\cC_{\alpha_1 \Delta }\,|
<
({\Delta}/{3})\,d^{\,\alpha_1 \Delta}
}
\end{gather*}
\begin{multline}
h_G \left( \,\cB_{G_{-\cC_{\alpha_1 \Delta -\ell} }}\left(p,r_{p, \alpha_1 \Delta -\ell } \right)  \, \right)
=
\frac
 {
   \left| \,\partial_G \left( \cB_{G_{-\cC_{\alpha_1 \Delta -\ell} }}\left(p,r_{p,\alpha_1 \Delta -\ell } \right)  \right) \, \right|
 }
 {
   \left| \,\cB_{G_{-\cC_{\alpha_1 \Delta -\ell} }}\left(p,r_{p, \alpha_1 \Delta -\ell }\right)  \, \right|
 }
\\
<
\frac 
  {
    ({\Delta}/{3})\,d^{\,\alpha_1 \Delta}
	}
	{
    \bee^{\nicefrac{\Delta^\mu}{8} }
	}
=
\frac 
  {
    ({\Delta}/{3})\,d^{\,\frac{\Delta^\mu}{14\,\log (2 d)} }
	}
	{
    \bee^{\nicefrac{\Delta^\mu}{8} }
	}
<
\frac 
  {
    ({\Delta}/{3})\,\left( d^{\nicefrac{1}{\log d}} \right) ^ {\,\Delta^\mu / 14}
	}
	{
    \bee^{\nicefrac{\Delta^\mu}{8} }
	}
\\
=
\frac 
  {
    ({\Delta}/{3})\, 2 ^ {\,\Delta^\mu / 14}
	}
	{
    \bee^{\nicefrac{\Delta^\mu}{8} }
	}
<
\frac 
  {
    {\Delta}/{3}  
	}
	{
    2 ^ { \Delta^\mu / 20 }
	}
<
\left( \frac{1}{\Delta} \right)^{1-\mu} \hspace*{-0.2in},
\text{\small since $\mu>0$ and $\Delta=\omega(1)$}
\label{easycase1}
\end{multline}
Inequality~\eqref{easycase1} implies that there is a set of 
$1+({\alpha_1 \Delta })/{2}=1+({\Delta^\mu})/({28\log(2d)})>{\Delta^\mu}/({56\log d})$
subsets of nodes 
$\cB_{G_{-\cC_{\alpha_1 \Delta} }}\left(p,r_{p, \alpha_1 \Delta }\right)\subset\cB_{G_{-\cC_{\alpha_1 \Delta -1} }}\left(p,r_{p, \alpha_1 \Delta -1 }\right)\subset
\dots\subset 
\cB_{G_{-\cC_{\nicefrac{\alpha_1 \Delta}{2} } }}\left(p,r_{p, \nicefrac{\alpha_1 \Delta}{2} }\right)$
such that each such subset 
$\cB_{G_{-\cC_{\alpha_1 \Delta -\ell} }}\left(p,r_{p, \alpha_1 \Delta -\ell }\right)$
has 
$h_G \left( \,\cB_{G_{-\cC_{\alpha_1 \Delta -\ell} }}\left(p,r_{p, \alpha_1 \Delta -\ell }\right)  \, \right)
< \left( \nicefrac{1}{\Delta} \right) ^{1-\mu}$. This proves our claim.

\medskip
\noindent
{\bf Case 2.2 (the difficult case)}:
Case $2.1$ does not hold.

\smallskip
This means that there exists an index $0\leq t\leq ({\alpha_1 \Delta})/{2}$ such that 
the removal of the set of nodes in 
$\cC_{\alpha_1 \Delta-t}$
does \emph{not} disconnect $p$ from $q$ in the corresponding graphs 
$G_{-\cC_{\alpha_1 \Delta-t} }$. This implies 
$r_{p,\alpha_1 \Delta-t}>\dist_{G_{-\cC_{\alpha_1 \Delta-t} } }(p,q)/2$.
For notational convenience, we will denote 
$\cC_{\alpha_1 \Delta-t}$ and 
$G_{-\cC_{\alpha_1 \Delta-t} }$ simply by 
$\cC$ and $G_{-\cC}$, respectively. We redefine
$\alpha_0=\alpha_1-({t}/{\Delta})$ such that
$\alpha_1 \Delta-t=\alpha_0 \,\Delta$. Note that ${\alpha_1}/{2}\leq\alpha_0\leq\alpha_1$.

\medskip
\noindent
{\bf First goal: 
show that our selection of $\alpha_0$ ensures
that removal of nodes in $\cC$ does not decrease the expansion of the balls
$\cB_{G_{-\cC}}(p,r)$
in the new graph $G_{-\cC}$
by more than a constant factor}.

\smallskip
First, note that the goal is trivially achieved if 
If $r\leq({\Delta}/{3})-\alpha_0\Delta$ since
for all $r \leq ({\Delta}/{3})-\alpha_0 \Delta$ we have 
$h_{G_{-\cC}}\left( \cB_{G_{-\cC}}(p,r) \right)=h_G\left( \cB_{G_{-\cC}}(p,r) \right)$.
Thus, assume that 
$r>({\Delta}/{3})-\alpha_0\Delta$.
To satisfy our goal, it {\em suffices} if we can show the following assertion: 
\begin{multline}
\forall\, ({\Delta}/{3})-\alpha_0 \Delta<r\leq \dist_{G_{-\cC}(p,q)}/2\,\, \colon
h_{G} \left( \cB_{G_{-\cC}}(p,r-1) \right)
> 
\left(\nicefrac{1}{\Delta}\right)^{1-\mu}
\,\,\Rightarrow
\\
h_{G_{-\cC}} \left( \cB_{G_{-\cC}}(p,r-1) \right) \,\geq\, 
h_{G} \left( \cB_{G_{-\cC}}(p,r-1) \right) /2 
\label{asser}
\end{multline}
We verify~\eqref{asser} as shown below. First, note that:
%
%
%
%
%
%
%
%
%
%
\begin{eqnarray}
& 
& 
h_{G_{-\cC}}\left( \cB_{G_{-\cC}}(p,r-1) \right)
\,\geq\,
h_{G} \left( \cB_{G_{-\cC}}(p,r-1) \right) / 2
\nonumber
\\
[2pt]
&
\equiv
&
\!
\frac{
\left| \,\partial_{G} \left( \cB_{G_{-\cC}}(p,r-1) \right) \, \right|
- \left| \, \partial_{G} \left( \cB_{G_{-\cC}}(p,r-1) \right) \cap\cC\, \right|
}
{
\left| \,\cB_{G_{-\cC}}(p,r-1)\, \right|
}
\geq
\frac { h_{G} \left( \cB_{G_{-\cC}}(p,r-1) \right) } {  2 }
\nonumber
\\
[2pt]
&
\Leftarrow
&
\frac{
\left| \,\partial_{G} \left( \cB_{G_{-\cC}}(p,r-1) \right) \, \right|
\,-\, |\cC|
}
{
\left| \,\cB_{G_{-\cC}}(p,r-1)\, \right|
}
\geq
h_{G} \left( \cB_{G_{-\cC}}(p,r-1) \right) / 2
\nonumber
\\
[2pt]
&
\equiv
&
\frac{
\left| \,\partial_{G} \left( \cB_{G_{-\cC}}(p,r-1) \right) \, \right|
}
{
\left| \,\cB_{G_{-\cC}}(p,r-1)\, \right|
}
\,-\,
\frac{
|\cC|
}
{
\left| \,\cB_{G_{-\cC}}(p,r-1)\, \right|
}
\geq
h_{G} \left( \cB_{G_{-\cC}}(p,r-1) \right) / 2
\nonumber
\\
[2pt]
& 
\equiv 
&
h_G\left( \cB_{G_{-\cC}}(p,r-1) \right)
\,-\,
\frac{
|\cC|
}
{
\left| \,\cB_{G_{-\cC}}(p,r-1)\, \right|
}
\geq
h_{G} \left( \cB_{G_{-\cC}}(p,r-1) \right) / 2
\nonumber
\\
[2pt]
& 
\equiv
& 
\frac{2\,|\cC|}{
h_{G} \left( \cB_{G_{-\cC}}(p,r-1) \right)
}
\leq
\left| \,\cB_{G_{-\cC}}(p,r-1)\, \right|
\nonumber
\\
[2pt]
& 
\Leftarrow
& 
\frac{2\, |\cC|}{
h_{G} \! \left( \cB_{G_{-\cC}}(p,r-1) \right)
}
\!
\leq
\!
\bee ^{\frac{\Delta^\mu}{8} },
\,
\text{since}
\left| \cB_{G_{-\cC}}(p,r-1) \right|
\!
\geq 
\!
\left| \cB_{G_{-\cC}}\! \left( \! p, \frac{\Delta}{3}-\alpha_0 \Delta \! \right) \right|
\nonumber
\\
&
&
\hspace*{2.0in}
=
\left| \,\cB_{G}\left( p, ({\Delta}/{3})-\alpha_0 \Delta \right)\, \right|
\nonumber
\\
&
&
\hspace*{2.0in}
\geq
\left| \,\cB_{G}\left( p, ({\Delta}/{3})-\alpha_1 \Delta \right)\, \right|
>
\bee ^{\nicefrac{\Delta^\mu}{8} }
\nonumber
\\
[2pt]
&
\Leftarrow
&
\left( ({\Delta}/{3})\,d^{\,\alpha_0 \Delta} \right)
\left( {2 } / \left({ h_{G} \left( \cB_{G_{-\cC}}(p,r-1) \right) } \right)
\right)
\leq
\bee ^{\nicefrac{\Delta^\mu}{8} },
\text{since}
\,
|\cC| \! <  \! ({\Delta}/{3})\, d^{\,\alpha_0 \Delta}
\nonumber
\\
[2pt]
&
\equiv
&
({\Delta^\mu}/{8})
\geq 
\ln \Delta + \alpha_0\,\Delta\ln d - \ln ({3}/{2}) - \ln \left( h_{G} \left( \cB_{G_{-\cC}}(p,r-1) \right) \right)  
\nonumber
\\
[2pt]
&
\Leftarrow
&
\frac{\Delta^\mu}{8}
\geq 
\ln \Delta + \alpha_1\,\Delta\ln d - \ln {\textstyle\frac{3}{2}} - \ln \left( h_{G} \left( \cB_{G_{-\cC}}(p,r-1) \right) \right),
\text{since}
\alpha_0 \leq \alpha_1
\nonumber
\\
[2pt]
&
\Leftarrow
&
\alpha_1 \leq 
\frac{
({\Delta^\mu}/{8})
- \ln \Delta + \ln \left( h_{G} \left( \cB_{G_{-\cC}}(p,r-1) \right) \right)
}
{ \Delta \ln d }
\label{eq1-new}
\end{eqnarray}
Now, if $h_{G} \left( \cB_{G_{-\cC}}(p,r-1) \right) >\left(\nicefrac{1}{\Delta}\right)^{1-\mu}$ then
since $\Delta=\omega(1)$ we have:  
\begin{gather*}
\frac{\Delta^\mu}{8} - \ln \Delta + \ln \left( h_{G} \left( \cB_{G_{-\cC}}(p,r-1) \right) \right)
>
{
\frac{\Delta^\mu}{8} - \ln \Delta -(1-\mu)\ln \Delta
>
({{\Delta}^\mu}/{7})
}
\end{gather*}
Thus, Inequality~\eqref{eq1-new} is satisfied by our selection
of $\alpha_1={1}/\left({14\,{\Delta}^{1-\mu}\,\log (2 d)}\right)$.
This verifies~\eqref{asser} and satisfies our first goal.

\medskip
\noindent
{\bf Second goal: Use the first goal and the fact that $\dist_{G_{-\cC}(p,q)}$ is large enough to find the desired subsets.}

\smallskip
First assume that there exists a set of $t=\max\left\{1,\,{D^\mu}/({56 \log d})\right\}$
indices 
$i_1<i_2<\dots<i_{t}$ in 
$\left\{ \frac{\Delta}{3} - \alpha_0 \Delta+1,\frac{\Delta}{3} - \alpha_0 \Delta+2,\dots,({\dist_{G_{-\cC}}(p,q)})/{2} \right\}$
such that 
\begin{equation}
\forall \, 1\leq s\leq t \, \colon \,\,
h_G\left( \, \cB_{G_{-\cC}}(p,i_s)\, \right) 
\leq \left(\nicefrac{1}{\Delta}\right)^{1-\mu}
\label{easy1}
\end{equation}
Obviously, the existence of these subsets
$\cB_{G_{-\cC}}(p,i_1)\subset\cB_{G_{-\cC}}(p,i_2)\subset\dots\subset\cB_{G_{-\cC}}(p,i_{t})$
proves our claim. 
Otherwise,
there are no sets of 
$t$ indices that satisfy~\eqref{easy1}.
This implies that there exists a set of 
$
\xi=
\left( 
{\dist_{G_{-\cC}}(p,q)}/{2}
\right)
-
\left( ({\Delta}/{3}) - \alpha_0 \Delta \right)
-
\left( t -1 \right)
$
distinct indices 
$j_1,j_2,\dots,j_{\xi}$ in 
\\
$\left\{ ({\Delta}/{3}) - \alpha_0 \Delta+1,({\Delta}/{3}) - \alpha_0 \Delta+2,\dots,({\dist_{G_{-\cC}}(p,q)})/{2} \right\}$
such that
\begin{multline}
\forall \, 1\leq s\leq \xi  \, \colon \,\,
h_G\left( \, \cB_{G_{-\cC}}(p,j_s)\, \right) 
> \left(\nicefrac{1}{\Delta}\right)^{1-\mu}
\,
\underset{\text{\small by \eqref{asser}} }
{
\Rightarrow
}
\\
\forall \, 1\leq s\leq \xi  \, \colon \,\,
h_{G_{-\cC}}\left( \, \cB_{G_{-\cC}}(p,j_s)\, \right) 
\geq 
h_G\left( \, \cB_{G_{-\cC}}(p,j_s)\, \right) /2
\label{kkkk}
\end{multline}
This in turn implies
\begin{eqnarray}
&
&
\left| \, \cB_{G_{-\cC}} \left( p, {\dist_{G_{-\cC}}(p,q)}/{2}  \right) \, \right|
\nonumber
\\
&
>
&
\left(
\hspace*{-0.5in}
\prod\limits_{j=0}^{\hspace*{0.5in}({\Delta}/{3})-\alpha_0 \Delta -1} 
     \hspace*{-0.55in} \left( 1+ h_G \left(  \cB_{G_{-\cC}} ( p, j ) \, \right) \,\right)
\right)
\underset{ \text{using \eqref{kkkk}} }
{
\left(
\hspace*{-0.7in}
\prod\limits_{\hspace*{0.6in}j=({\Delta}/{3})-\alpha_0 \Delta }^{\hspace*{0.7in}({\Delta}/{3})-\alpha_0 \Delta+\xi-1} 
     \hspace*{-0.7in} \left( 1+  \left( h_G \left(  \cB_{G_{-\cC}} ( p, j ) \, \right) /2 \right) \right)
\right)
}
\nonumber
\\
&
>
&
\left(
\hspace*{-0.5in}
\prod\limits_{j=0}^{\hspace*{0.5in}({\Delta}/{3})-\alpha_0 \Delta -1} 
     \hspace*{-0.55in} \bee^ { h_G (  \cB_{G_{-\cC}} ( p, j ) \, )\,/\,2 }
\right)
\left(
\hspace*{-0.7in}
\prod\limits_{\hspace*{0.6in}j=({\Delta}/{3})-\alpha_0 \Delta }^{\hspace*{0.7in}({\Delta}/{3})-\alpha_0 \Delta+\xi-1} 
     \hspace*{-0.7in} \bee^ { h_G (  \cB_{G_{-\cC}} ( p, j ) \, ) \,/\,4 }
\right)
\nonumber
\\
[2pt]
&
=
&
\Bigg(
 \bee^ { \hspace*{-0.45in} \sum\limits_{j=0}^{\hspace*{0.5in} ({\Delta}/{3})-\alpha_0 \Delta -1} \hspace*{-0.55in} h_G (  \cB_{G_{-\cC}} ( p, j ) \, \,)/2 }
\Bigg)
 \,\,
\Bigg(
 \bee^ { \hspace*{-0.6in} \sum\limits_{\hspace*{0.5in} j=({\Delta}/{3})-\alpha_0 \Delta}^{\hspace*{0.6in} ({\Delta}/{3})-\alpha_0 \Delta+\xi-1} \hspace*{-0.65in} h_G (  \cB_{G_{-\cC}} ( p, j ) \, ) / 4 }
\Bigg)
>
 \bee^ { \hspace*{-0.6in} \sum\limits_{j=0}^{\hspace*{0.6in} ({\Delta}/{3})-\alpha_0 \Delta+\xi-1} \hspace*{-0.65in} h_G (  \cB_{G_{-\cC}} ( p, j ) \, ) /4 }
\label{sum-bound}
\end{eqnarray}
Using \eqref{sum-bound} and our specific choice of the node $p$ (over node $q$), we have 
\begin{multline}
{n}/{2} > 
\left| \, \cB_{G_{-\cC}}\left(p,{\dist_{G_{-\cC}}(p,q)}/{2}\right) \, \right|
>
\bee^ { \hspace*{-0.6in} \sum\limits_{j=0}^{\hspace*{0.6in} ({\Delta}/{3})-\alpha_0 \Delta+\xi-1} \hspace*{-0.65in} h_G (  \cB_{G_{-\cC}} ( p, j ) \, ) /4 }
\Rightarrow
\\
 \hspace*{-0.7in}
 \sum\limits_{j=0}^{\hspace*{0.6in} ({\Delta}/{3})-\alpha_0 \Delta+\xi-1} \hspace*{-0.7in} h_G (  \cB_{G_{-\cC}} ( p, j ) \, )
< 4 \ln n
\label{sum-bound-2}
\end{multline}
We now claim 
that there \emph{must} exist a set of $t={\Delta^\mu}/({56 \log d})$
distinct indices 
$i_1<i_2<\dots<i_{t}$ in
$\left\{ 0,1,\dots,({\Delta}/{3})-\alpha_0 \Delta+\xi-1 \right\}$
such that 
\begin{equation}
\forall \, 1\leq s\leq t \, \colon \,\,
h_G\left( \, \cB_{G_{-\cC}}(p,i_s)\, \right) 
\leq 
({500\ln n})/\left({\Delta\,2^{ \, {\Delta^{\mu}}/ ( {28 \delta \log(2d)} ) } } \right)
\label{easy2}
\end{equation}
The existence of these indices will obviously prove our claim. Suppose, for the sake of contradiction, 
that this is \emph{not} the case.
Together with \eqref{sum-bound-2} this implies:
\begin{eqnarray}
& & 
\hspace*{-0.05in}
4 \ln n 
 \,>\,\,\,
 \hspace*{-0.7in}
 \sum\limits_{j=0}^{\hspace*{0.6in} ({\Delta}/{3})-\alpha_0 \Delta+\xi-1} \hspace*{-0.6in} h_G (  \cB_{G_{-\cC}} ( p, j ) \, )
\nonumber
\\
& & 
   \underset{\hspace*{0.15in}\text{\footnotesize by \eqref{easy2}} }
   {
     \hspace*{0.15in} > 
   }
\!\!\!
 \left( \frac{\Delta}{3}-\alpha_0 \Delta+\xi - \frac{\Delta^\mu}{56 \log d} +1 \right)
    \left( ({500\ln n}) / \left( {\Delta\,2^{ \, {\Delta^{\mu}} / ( {28 \delta \log(2d)} ) \, } } \right) \, \right)
\nonumber
\\
[2pt]
& & 
\hspace*{-0.05in}
\Rightarrow 
 \left( 
{\textstyle\frac{\dist_{G_{-\cC}}(p,q)}{2}}
 - \max \left\{ 1, \, {\textstyle \frac {\Delta^\mu} {28 \log d} } \right\} \right)
    \left( ({500\ln n}) / \left( {\Delta\,2^{ \, \frac {\Delta^{\mu}} {28 \delta \log(2d)} \, } } \right) \, \right)
 < 
 4 \ln n,
\nonumber
\\
& & 
\hspace*{2.75in}
\text{\small substituting the values of $t$ and $\xi$}
\nonumber
\\
[2pt]
& & 
\hspace*{-0.05in}
\Rightarrow
 \left( 
  \frac {\Delta}{120}\,2^{ \, \frac {\alpha_1 \Delta} {2\delta} }
 - 
  \max \left\{ 1, \, \frac {\Delta^\mu} {28 \log d} \right\} \,\right)
    \left( ({500\ln n}) / \left( {\Delta\,2^{ \, \frac {\Delta^{\mu}} {28 \delta \log(2d)} \, } } \right) \, \right)
 < 
 4 \ln n,
\nonumber
\\
& & 
\hspace*{3in}
\text{\small by \eqref{eq:a2} and since ${\alpha_1}/2\leq\alpha_0$}
\nonumber
\\
[2pt]
& & 
\hspace*{-0.05in}
\equiv
 \left( 
  \frac {\Delta}{120}\,2^{ \, \frac {\Delta^{\mu}} {28 \delta \log(2d)} }
 - 
  \max \left\{ 1, \, \frac {\Delta^\mu} {28 \log d} \right\} \,\right)
    \left( 125 / \left( {\Delta\,2^{ \, {\Delta^{\mu}} / ( {28 \delta \log(2d)} ) \, } } \right) \, \right)
 < 
 1 
\nonumber
\\
[2pt]
& & 
\hspace*{-0.05in}
\Rightarrow
 \left( 
  ({\Delta}/{121})\,2^{ \, {\Delta^{\mu}} / ( {28 \delta \log(2d)} ) }
  \right)
    \left( 125 / \left( {\Delta\,2^{ \, {\Delta^{\mu}} / ( {28 \delta \log(2d)} ) \, } } \right) \, \right)
 < 
 1 
\nonumber
\\
[2pt]
& & 
\hspace*{-0.05in}
\equiv
\nicefrac{125}{121} < 1, 
\,\,
\text{\small since $\Delta=\omega(1)$}
\label{final}
\end{eqnarray}
Since~\eqref{final} is false,
there 
must exist a set of $t$
distinct indices 
$i_1<i_2<\dots<i_t$
such that~\eqref{easy2} holds and the corresponding sets 
$\cB_{G_{-\cC}}\left(p,i_1\right)\subset\cB_{G_{-\cC}}\left(p,i_2\right)\subset\dots\subset\cB_{G_{-\cC}}\left(p,i_t\right)$
prove our claim.

\bigskip
\noindent
(III)
{\bf Time complexity for finding each witness}

\medskip
It should be clear that we can find each witness provided we can implement the 
following steps:
\begin{itemize}[leftmargin=*]
\item
Find two nodes $p$ and $q$ such that $\dist_G(p,q)=\Delta$ in $O\left(n^2\log n + mn\right)$ time.
\item
Using breadth-first-search (BFS), find the two nodes $p',q'$ as in the proof in $O(m+n)$ time.
\item
There are at most $\alpha_1\Delta/2=\Delta^\mu/(28 \, \log(2d)\,)<n$ possible values of $\alpha$ 
considered in the proof. For each $\alpha$, the following steps are needed: 
\begin{itemize}
\item
Use BFS find the set of nodes $\cC$ in $O\left(n^2+mn\right)$ time.
\item
Compute $G_{-\cC}$ in $O(m+n)$ time.
\item
Use BFS to compute $\cB_{G_{-\cC}}(p,r)$ for every $0\leq r\leq \dist_{G_{-\cC}}(p,q)/2$ in $O(m+n)$ time.
\item
Compute $h_G\left(\cB_{G_{-\cC}}(p,r)\right)$ for every $0\leq r\leq \dist_{G_{-\cC}}(p,q)/2$ 
in $O(n^2+mn)$ time, and select a subset of nodes with a minimum expansion.
\end{itemize}
\end{itemize}
%


\subsection{Family of Witnesses of Node/Edge Expansion With Limited Mutual Overlaps}

The result in the previous section provided a nested family of cuts of small expansion that separated node $p$ from node $q$.
However, pairs of subsets in this family may differ by as few as just {\em one} node.
In some applications, one may need to generate a family of cuts that are sufficiently different from each other, \IE,
they are either disjoint or have limited overlap.
The following theorem addresses this question.

\begin{theorem}
\label{main3}
Let $p$ and $q$ be any two nodes of $G$ and let $\Delta=\dist_G(p,q)>8$.
Then, for any constant $0<\mu<1$ and for any positive integer 
$\tau < {\Delta} / \left( { \big( 42\,\delta \, \log (2d) \, \log ( 2 \Delta ) \,\big)^{{1}/{\mu}} } \right)$
the following results hold for $\langle G,d,\delta\rangle$: 
there exists 
$\left\lfloor{\tau}/{4}\right\rfloor$ 
distinct collections of subsets of nodes 
$\emptyset\subset \cF_1,\cF_2,\dots,\cF_{\left\lfloor{\tau}/{4}\right\rfloor}\subset 2^V$
such that
\begin{enumerate}[label=$\blacktriangleright$]
\item
$\displaystyle
\forall\, j\in \left\{1,\dots, \left\lfloor {\tau}/{4}\right\rfloor \right\} \,
\forall\, S\in \cF_j
\,\colon
$

\hfill
$\displaystyle
h_G\left(S\right) 
\leq 
\max\left\{ \left(\frac{1}{(\Delta/\tau)}\right)^{1-\mu},\,
\dfrac { 360 \log n }{ (\Delta/\tau)\,2^{\,   \frac{ (\Delta/\tau)^{\mu}} {7\,\delta\,\log(2d) }    } }
\, \right\}
$.
\smallskip
\item
Each collection $\cF_j$ has at least 
$t_j=\max\left\{ \frac{(\Delta/\tau)^\mu}{56\log d},\,1 \right\}$ subsets
$V_{j,1},\dots,V_{j,t_j}$ that form a nested family, {\em\IE}, 
$\emptyset \subset V_{j,1}\subset V_{j,2}\subset\dots\subset V_{j,t_j}\subset V$.
\smallskip
\item
All the subsets in each $\cF_j$ can be found in a total of $O\left(n^3\log n + mn^2\right)$ time.
\smallskip
\item
{\rm (limited overlap claim)}
For every pair of subsets $V_{i,k}\in\cF_i$ and $V_{j,k'}\in\cF_j$ with $i\neq j$, 
either $V_{i,k}\cap V_{j,k'}=\emptyset$ or 
at least 
${\Delta}/({2\,\tau})$
nodes in each subset do not belong to the other subset.
\end{enumerate}
\end{theorem}

\begin{remark}
\label{rem4}
Consider a bounded-degree hyperbolic graph, \IE, assume that $\delta$ and $d$ are constants.
Setting $\tau={\Delta^{1/2}}$ gives $\Omega( {\Delta^{1/2}} )$ nested families 
of subsets of nodes, with each family having at least $\Omega( {\Delta^{1/2}} )$ subsets each of maximum node
expansion $\left( {1}/{\Delta} \right)^{ ({1-\mu})/{2}}$, such that every pairwise non-disjoint subsets
from different families have at least $\Omega( {\Delta^{1/2}} )$ private nodes.
\end{remark}

\begin{proof}
Select $\tau\leq \Delta/4$ such that $\tau$ satisfies the following:
\begin{gather}
{\Delta} / ( {60\,\tau} ) \,2^{\, ( \, ( \Delta / \tau )^\mu \,  ) / ({28\,\delta \, \log (2d)})  }
>
( {\Delta} / {\tau} ) + 2 \, \Delta
\label{tau-bound:eq}
\end{gather}
Note that 
$\tau \geq { \left( 42\,\delta \, \log (2d) \, \log ( 2 \Delta ) \,\right)^{{1}/{\mu}} } / {\Delta}$
satisfies~\eqref{tau-bound:eq} since 
\begin{multline*}
{\Delta} / ( {60\,\tau} ) \,2^{\, ( \, ( \Delta / \tau )^\mu \,  ) / ({28\,\delta \, \log (2d)})  }
>
({\Delta}/{\tau}) + 2 \,\Delta
\\
\Rightarrow\,
(\Delta / \tau )^\mu
>
28\,\delta \, \log (2d) \, \log ( 60+120\,\tau )
\underset{ \hspace*{-1.2in} \text{\small since $\tau < {\Delta}/{4}$} }
{
>
168\,\delta \, \log (2d) \, \log ( 2 \Delta )
}
\\
\,\Rightarrow \,
\tau < 
{\Delta}
/
\left(
{
\big( 42\,\delta \, \log (2d) \, \log ( 2 \Delta ) \,\big)^{{1}/{\mu}}
}
\right)
\end{multline*}
Let $\left(p=p_1,p_2,\dots,p_{\tau+1}=q\right)$ be an ordered sequence of $\tau+1$ nodes 
such that $\dist_G\left(p_i,p_{i+1}\right)={\Delta}/{\tau}$ for $i=1,2,\dots,\tau$.
Applying Theorem~\ref{main-nested} for each pair $\left(p_i,p_{i+1}\right)$,
we get 
a nested family  
$\emptyset\subset \cF_i \subset 2^V$
of subsets of nodes 
such that 
$t_i=|\,\cF_i\,| \geq \max\left\{ \frac{(\Delta/\tau)^\mu}{56\log d},\,1 \right\}$ and, for any 
$V_{i,k}\in\cF_i$, 
$
h_G\left(V_{i,k}\right) 
\leq 
\max\left\{ \left( {1}/ ({\Delta/\tau )}\right)^{1-\mu},\,
{ ( 360 \log n ) } / \left( { (\Delta/\tau)\,2^{\,    { (\Delta/\tau)^{\mu}} / \left( {7\,\delta\,\log(2d) } \right)   } } \right)
\, \right\}
$.
Recall that 
the subset of nodes $V_{i,k}$
was constructed in Theorem~\ref{main-nested} in the following manner
(see \FI{p2-fig} for an illustration):
\begin{itemize}
\item
Let $\ell_i$ and $r_i$ be two
nodes on a shortest path $\overline{p_i,p_{i+1}}$ such that 
$
\dist_G\left(p_i,\ell_i\right)
=
\dist_G\left(\ell_i,r_i\right)
=
\dist_G\left(r_i,p_{i+1}\right)
=
\dist_G\left(p_i,p_{i+1}\right) / 3
$.
\smallskip
\item
For some 
$
{1} / \left( {28\,(\Delta/\tau)^{1-\mu}\log (2d) } \right)
\leq 
\alpha_{i,k} 
\leq 
{1} / \left( {14\,(\Delta/\tau)^{1-\mu}\log (2d) } \right)
<1/4
$, 
construct the graph
$G_{-\cC_{i,k}}$ obtained
by removing the set of nodes 
$\cC_{i,k}$ which are \emph{exactly} at a distance of
$\lceil \alpha_{i,k}\,\dist_G\left(p_i,p_{i+1}\right) \rceil$
from some node of the shortest path $\overline{\ell_i,r_i}$. 
\smallskip
\item
The subset
$V_{i,k}$ is then the ball
$
\cB_{G_{-\cC_{i,k}}}\left(y_i,a_{i,k}\right)
$
for some
\\
$
a_{i,k}\in \big[0,\,{\dist_{G_{-\cC_{i,k}}}\left(p_i,p_{i+1}\right)}/{2} \big]
$
and for some $y_i\in\left\{p_i,p_{i+1}\right\}$.
If $y_i=p_i$ then we call the collection of subsets $\cF_i$ ``left handed'', otherwise we call $\cF_i$ ``right handed''.
\end{itemize}
We can partition the set of $\tau$ collections $\cF_1,\dots,\cF_{\tau}$
into four groups depending on whether the subscript $j$ of $\cF_j$ is odd or even, and 
whether $\cF_j$ is left handed or right handed. One of these $4$ groups must
at least  $\left\lfloor{\tau}/{4}\right\rfloor$ family of subsets.
Suppose, without loss of generality, that this happens for the collection of families
that contains $\cF_{i,k}$ when $i$ is even and $\cF_{i,k}$ is left handed (the other cases are similar).
We now show that subsets in this collection that belong to different families
do satisfy the \emph{limited overlap} claim.

\begin{figure*}[htbp]
\includegraphics{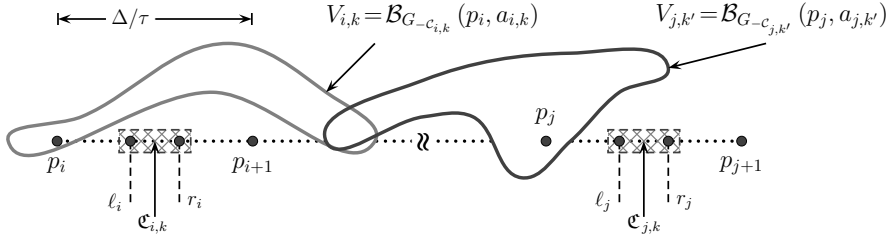}
\caption{\label{p2-fig}Illustration of various quantities related to the proof of Theorem~\ref{main3}.
Nodes within the lightly cross-hatched region belong to 
$\mathfrak{C}_{i,k}$
and
$\mathfrak{C}_{j,k'}$.
Note that 
$\cB_{G_{-\cC_{i,k}}}\left(p_i,a_{i,k}\right)$ and $\cB_{G_{-\cC_{j,k'}}}\left(p_j,a_{j,k'}\right)$
need not be balls in the original graph $G$.
}
\end{figure*}

Consider an arbitrary set in the above-mentioned collection of the form 
$V_{i,k} = \cB_{G_{-\cC_{i,k}}}\left(p_i,a_{i,k}\right)$
with even $i$. 
Let $\mathfrak{C}_{i,k}$ denote the nodes in the \emph{interior} of the closed cylinder of nodes in $G$ 
which are at a distance of \emph{at most}
$\left\lceil\alpha_{i,k}\,\dist_G\left(p_i,p_{i+1}\right)\right\rceil$
from some node of the shortest path $\overline{\ell_i,r_i}$, \IE,
let
$
\mathfrak{C}_{i,k} = \left\{ u \,|\, \exists \, v \in \overline{\ell_i,r_i} 
\colon 
\dist_G(u,v)\leq \left\lceil \alpha_{i,k}\,\dist_G\left(p_i,p_{i+1}\right) \right\rceil \, \right\} 
$
(see \FI{p2-fig}).
Let 
$V_{j,k'} = \cB_{G_{-\cC_{j,k'}}}\left(p_j,a_{j,k'}\right)$
be a set in another family $\cF_j$ with even $j\neq i$ (see \FI{p2-fig}).
Assume, without loss of generality, that $i$ is smaller than $j$, \IE, 
$i\leq j-2$ (the other case is similar).

\begin{proposition}
\label{p1}
$\mathfrak{C}_{i,k} \cap \cB_{G_{-\cC_{j,k'}}}\left(\,p_j, {\Delta}/ ({2\,\tau}) \,\right) =\emptyset$. 
\end{proposition}

\begin{proof}
Assume for the sake of contradiction that 
$u\in\mathfrak{C}_{i,k} \,\cap \,\cB_{G_{-\cC_{j,k'}}}\left(\,p_j, {\textstyle \frac {\Delta} {2\,\tau} } \,\right)$.
Since 
$u\in \mathfrak{C}_{i,k}$, there is a node $v\in\overline{\ell_i,r_i}$ such that 
\[
\dist_G(v,u)\leq \lceil\alpha_{i,k}\,\dist_G\left(p_i,p_{i+1}\right)\rceil
<
\dist_G\left(p_i,p_{i+1}\right)/4
=
{\Delta} / ({4\,\tau})
\]
Thus, 
\begin{multline*}
u\in \cB_{G_{-\cC_{j,k'} }}\left(p_j, {\textstyle \frac{\Delta} {2\,\tau} } \right)
\,\Rightarrow\,
\dist_{G_{-\cC_{j,k'}}} \left(u,p_j \right) \leq {\textstyle \frac {\Delta} {2\,\tau} }
\,\Rightarrow\,
\dist_{G} \left(u,p_j \right) \leq {\textstyle \frac {\Delta} {2\,\tau} }
\\
\Rightarrow \,
\dist_G\left(v,p_j\right) \leq 
\dist_G\left(v,u\right) + \dist_G\left(u,p_j\right)
<
{\Delta} / ({4\,\tau}) + {\Delta} / ({2\,\tau} )
<
{\Delta} / {\tau}
\end{multline*}
which contradicts the fact that
$
\dist_G\left(v,p_j\right)>
\dist_G\left(p_{i+1},p_j\right)
= {\Delta} / {\tau}
$.
\end{proof}

\begin{proposition}
\label{p2}
$\dist_{ \,G_{-\cC_{j,k'}} }\left(u,p_j\right)>{\Delta} / ({2\,\tau})$ for 
any node 
$
u\in V_{i,k}\cap V_{j,k'}
=
\cB_{G_{-\cC_{i,k}}}\left(p_i,a_{i,k}\right)
\cap
\cB_{G_{-\cC_{j,k'}}}\left(p_j,a_{j,k'}\right)
$.
\end{proposition}

\begin{proof}
Assume for the sake of contradiction that $z=\dist_{ \,G_{-\cC_{j,k'}} }\left(u,p_j\right)\leq {\Delta} / ({2\,\tau})$.
Since $u\in V_{i,k}  = \cB_{G_{-\cC_{i,k}}} \left(p_i,a_{i,k}\right)$, this implies 
\[
\dist_{ \,G_{-\cC_{i,k}} }\left(p_{i,k},u\right)\leq a_{i,k} \leq 
{\dist_{G_{-\cC_{i,k}}}\left(p_i,p_{i+1}\right)} / {2}
\]
Since $u\in V_{j,k'} \! = \! \cB_{G_{-\cC_{j,k'}}} \!\! \left(p_j,a_{j,k'}\right)$ for some $a_{j,k'}$, this implies 
$u\in \cB_{G_{-\cC_{j,k'}}} \!\! \left(p_j,z\right)$.
Since 
$z\leq {\Delta} / ({2\,\tau})$, 
by Proposition~\ref{p1} $\mathfrak{C}_{i,k}\cap \cB_{G_{-\cC_{j,k'}}}\left(p_j,z\right)=\emptyset$, 
and therefore
\[
{\Delta} / ( {2\,\tau} ) \geq 
\underset{\text{\small since 
$\mathfrak{C}_{i,k}\cap \cB_{G_{-\cC_{j,k'}}}\left(p_j,z\right)=\emptyset$
} }
{
z=\dist_{ \,G_{-\cC_{j,k'}} }\left(u,p_j\right)
=
\dist_{ \,G_{-\cC_{i,k}\cup \,\cC_{j,k'}} }\left(u,p_j\right)
}
\geq
\dist_{ \,G_{-\cC_{i,k}} }\left(u,p_j\right)
\]
which in turn implies
\begin{multline}
\dist_{ \,G_{-\cC_{i,k}} }\left(p_i,p_j\right)
\leq 
\dist_{ \,G_{-\cC_{i,k}} }\left(p_i,u\right)
+
\dist_{ \,G_{-\cC_{i,k}} }\left(u,p_j\right)
\\
\leq
{\dist_{G_{-\cC_{i,k}}}\left(p_i,p_{i+1}\right)} / {2}
+
{\Delta} / ({2\,\tau})
\label{ineq1}
\end{multline}
Since the Hausdorff distance between the two shortest paths $\overline{\ell_i,r_i}$ and 
$\overline{p_j,p_{j+1}}$ is at least 
$
(j-i-1)\frac{\Delta}{\tau}+ 
\Delta/(3\,\tau)
>
\alpha_{i,k}\,\dist_G\left(p_i,p_{i+1}\right)$ 
and 
$\dist_{ \,G_{-\cC_{i,k}} }\left(p_j,p_{i+1}\right)
=
(j-i){\Delta}/{\tau}
<\Delta
$,
we have
\begin{multline}
\dist_{ \,G_{-\cC_{i,k}} }\left(p_i,p_{i+1}\right)
\leq
\dist_{ \,G_{-\cC_{i,k}} }\left(p_i,p_j\right)
+
\dist_{ \,G_{-\cC_{i,k}} }\left(p_j,p_{i+1}\right)
\\
\underset{ \text{\small by \eqref{ineq1}} }
{
\leq
}
{\dist_{G_{-\cC_{i,k}}}\left(p_i,p_{i+1}\right)} / {2}
+
{\Delta} / ({2\,\tau})
+
\Delta
\\
\Rightarrow\,
\dist_{ \,G_{-\cC_{i,k}} }\left(p_i,p_{i+1}\right)
\leq
{\Delta}/{\tau} + 2\Delta
\label{gt:eq}
\end{multline}
On the other hand, by Fact~\ref{fact-cylinder}:
\begin{gather}
\dist_{ \,G_{-\cC_i} }\left(p_i,p_{i+1}\right)
\geq
{\Delta} / ({60\,\tau}) \,2^{\,({\alpha_{i,k} \, \Delta}) / ( {\delta \, \tau } )  }
\geq
{\Delta} / ( {60\,\tau} ) \,2^{\, ( \, ( \Delta / \tau )^\mu \,  ) / ({28\,\delta \, \log (2d)})  }
\label{less:eq}
\end{gather}
Inequalities~\eqref{gt:eq}~and~\eqref{less:eq} together imply
\begin{gather}
{\Delta} / ( {60\,\tau} ) \,2^{\, ( \, ( \Delta / \tau )^\mu \,  ) / ({28\,\delta \, \log (2d)})  }
\leq
( {\Delta} / {\tau} ) + 2 \, \Delta
\label{hhh}
\end{gather}
Inequality~\eqref{hhh} contradicts Inequality~\eqref{tau-bound:eq}.
\end{proof}

To complete the proof of limited overlap claim, suppose that $V_{i,k}\cap V_{j,k'}\neq\emptyset$ and let 
$u\in V_{i,k}\cap V_{j,k'}$.
Proposition~\ref{p2} implies that 
$V_{j,k'}\supset \cB_{ G_{-\cC_{j,k'}} } \left( p_j,{\Delta} / ( {2\,\tau} ) \right)$, 
$u\notin \cB_{ G_{-\cC_{j,k'}} } \left( p_j, {\Delta} / ( {2\,\tau} ) \right)$, 
and thus
there are at least ${\Delta} / ( {2\,\tau} )$ node on a shortest path 
in $G_{-\cC_{j,k'}}$ from $p_j$ to a node at a distance of ${\Delta} / ( {2\,\tau} )$ from $p_j$ that are
not in $V_{i,k}$.
\end{proof}

\subsection{Family of Mutually Disjoint Cuts}
\label{sec-disjoint}

Recall that, given two distinct nodes $s,t\in V$ of a graph $G=(V,E)$, a cut in $G$ that separates $s$ from $t$ 
(or, simply a ``$s$-$t$ cut'') $\cut_G(S,s,t)$ is a subset of nodes $S$ that disconnects $s$ from $t$.
The \emph{cut-edges} $\cE_G(S,s,t)$ 
(resp., \emph{cut-nodes} $\cV_G(S,s,t)$)
corresponding to this cut is the set of edges with one end-point in $S$
(resp., the end-points of these cut-edges that belong to $S$), \IE,
\begin{gather*}
\cE_G(S,s,t)
=
\left\{ \,
\{u,v\} \,|\, u\in S, \, v\in V\setminus S, \, \{u,v\}\in E 
\, \right\},\,
\\
\cV_G(S,s,t)
=
\left\{ \,
u \,|\, u\in S, \, v\in V\setminus S, \, \{u,v\}\in E 
\, \right\}
\end{gather*}
{\bf Note that in the following lemma 
$d$ is the maximum degree of any node ``except $s$, $t$ and any node within a distance of 
$35\,\delta$ of $s$''} (degrees of these nodes may be arbitrary).

\begin{lemma}\label{const-thm}
Suppose that the following holds for our given $\langle G,d,\delta\rangle$:
\begin{enumerate}[label=$\blacktriangleright$]
\item
$s$ and $t$ are two nodes of $G$ such that $\dist_G(s,t)>48\,\delta+8\,\delta\log n$, and
\smallskip
\item
$d$ is the maximum degree of any node except $s$, $t$ and any node within a distance of 
$35\,\delta$ of $s$ (degrees of these nodes may be arbitrary).
\end{enumerate}
Then, 
there exists a set of at least
$
\frac {\dist_G(s,t)-8\,\delta\log n} {50 \,\delta }
=
\Omega \left(
\dist_G(s,t)
\right)
$
(node and edge) disjoint cuts such that each such cut has at most 
$d^{\,12\delta+1}$ 
cut edges.
\end{lemma}

\begin{remark}
Suppose that $G$ is hyperbolic (\emph{\IE}, $\delta$ is a constant), $d$ is a constant,
and $s$ and $t$ be two nodes such that $\dist_G(s,t)>48\delta+8\delta\log n=\Omega(\log n)$.
Lemma~\ref{const-thm} then implies that
there are $\Omega \left(\dist_G(s,t)\right)$ $s$-$t$ cuts each having $O(1)$ edges.
If, on the other hand, $\delta=O(\log\log n)$, then such cuts have $\mathrm{polylog}(n)$ edges.
\end{remark}

\begin{remark}
The bound in Lemma~\ref{const-thm} is obviously meaningful only if $\delta=o(\log_d n)$.
If $\delta=\Omega(\log_d n)$, then $\delta$-hyperbolic graphs include expanders and thus many 
small-size cuts may not exist in general.
\end{remark}


\begin{proof}
Recall that we may assume that $\delta\geq {1}/{2}$.
We start by doing a BFS starting from node $s$. Let $\cL_i$ be the sets of nodes 
at the $i\tx$ level (\IE, $\forall\,u\in\cL_i\colon \dist_G\left(s,u\right)=i$); 
obviously $t\in\cL_{\dist_G(s,t)}$.
Assume $\dist_G(s,t)>48\,\delta+8\,\delta\,\log n$, and 
consider two arbitrary paths $\cP_1$ and $\cP_2$ between $s$ and $t$ passing through two nodes 
$v_1,v_2\in\cL_j$ for some $48\,\delta \leq j\leq \dist_G(s,t)-7\,\delta\,\log n$.

We first claim that $\dist_G\left(v_1,v_2\right)<12\,\delta$.
Suppose, for the sake of contradiction, suppose that 
$\dist_G\left(v_1,v_2\right)\geq 12\delta$.
Let $v_1'$ and $v_2'$ be the first node in level $\cL_{j+6\delta\log n}$ visited by $\cP_1$ and $\cP_2$, respectively. 
Since both $\cP_1$ and $\cP_2$ are paths between $s$ and $t$ and 
$j+6\,\delta\,\log n<\dist_G(s,t)$ implies 
$\cL_{j+6\delta\log n+1}\neq\emptyset$,
there must be a path $\cP_3$ between $v_1'$ and $v_2'$ through $t$ using nodes
not in 
$\bigcup_{\,0 \, \leq  \, \ell  \, \leq  \, j+6\,\delta\log n}\hspace*{-0.0in}\cL_{\ell}$.
We show that this is impossible by Fact~\ref{fact-expo}.
Set the parameters in Fact~\ref{fact-expo} in the following manner:
$\kappa=4$, $\alpha=6\,\delta\,\log n$, $r=j>12\,\kappa\,\delta=48\,\delta$; 
$u_1=v_1$, $u_2=v_2$, $u_4=v_1'$, and $u_3=v_2'$.
Then the length of $\cP_3$ satisfies
$
\left| \, \cP_3 \, \right|
>
2^{\log n+5} > n
$
which is impossible since $\left| \, \cP_3 \, \right|<n$.

We next claim that, for  
any arbitrary node in level $v\in\cL_j$ lying on a path between $s$ and $t$,  
$\cB_G\left(v,12\delta\right)$ provides an $s$-$t$ cut 
$\cut_G\left( \cB_G\left(v,12\delta\right),s,t \right)$ having at most 
$\cE_G\left( \cB_G\left(v,12\delta\right),s,t \right)\leq d^{\,12\delta+1} $ edges.
To see this, 
consider any path $\cP$ between $s$ and $t$ and let $u$ be the first node in $\cL_j$ visited by the path.
Then, 
$\dist_G(u,v)\leq 12\delta$ and thus $v\in\cB_G\left(v,12\delta\right)$. Since nodes in 
$\cB_G\left(v,12\delta\right)$
are at a distance of at least $35\delta$ from $s$ and 
$t\notin\cB_G\left(v,12\delta\right)$, 
$d$ is the maximum degree of any node 
in 
$\cB_G\left(v,12\delta\right)$ and 
it follows that 
$\cE_G\left( \cB_G\left(v,12\delta\right),s,t \right)\leq d \, \partial_G\left( \cB_G\left(v,12\delta-1\right) \right)
\leq d^{\,12\delta+1}$.

We can now finish the proof of our lemma in the following way.
Assume that $\dist_G(s,t)>48\, \delta + 8 \,\delta\,\log n$.
Consider the levels $\cL_j$ for 
$j\in \big\{\,50\delta,100\delta,150\delta,\dots,\allowbreak \frac {\dist_G(s,t)-8\delta\log n} {50\delta} \,\big\}$.
For each such level $\cL_j$, select a node $v_j$ that is on a path between $s$ and $t$ and consider the subset of edges 
in 
$\cut_G\big( \cB_G\big(v_j,12\delta\big),s,t \big)$.
Then, 
$\cut_G\left( \cB_G\left(v_j,12\delta\right),s,t \right)$ over all $j$ provides our family of $s$-$t$ cuts.
The number of such cuts is at least
$({\dist_G(s,t)-8\delta\log n}) / ({50\delta})$. To see why these cuts are node and edge disjoint, 
note that 
$\cE_G ( \cB_G(v_j,12\delta),s,t ) \,\cap\, \cE_G ( \cB_G(v_\ell,12\delta),s,t )=\emptyset$
and 
$\cV_G\left( \cB_G\left(v_j,12\delta\right),s,t \right)\cap \cV_G\left( \cB_G\left(v_\ell,12\delta\right),s,t \right)=\emptyset$
for any $j\neq\ell$ since $\dist_G\left(v_j,v_\ell\right)>30\delta$. 
\end{proof}


\section{Algorithmic Applications}
\label{sec-appl}

In this section, we consider a few algorithmic applications of the bounds and proof techniques we showed 
in the previous section.

\subsection{Network Design Application: Minimizing Bottleneck Edges}
\label{sec-bott}
\label{sec-size}

In this section we consider the following problem.

\begin{problem}[Unweighted Uncapacitated Minimum Vulnerability problem (\uumv)~\cite{n1,n2,n3}]
The input to this problem 
a graph $G=(V,E)$, two nodes $s,t\in V$, and two positive integers $0<r<\kappa$. 
The goal is to find a set of $\kappa$ paths between $s$ and $t$ 
that minimizes the number of ``shared edges'', where 
an edge is called shared if it is in more than $r$ of these $\kappa$ paths between $s$ and $t$. 
When $r=1$, the 
\uumv\ 
problem is called the ``minimum shared edges''
(\mse) 
problem. 
\end{problem}

We will use the notation 
$\OPT_{ \!\!\!\!\text{{\footnotesize\uumv}} \!\! }(G,s,t,r,{\kappa})$ 
to denote the number of shared edges in an optimal solution of an instance of \uumv.
\uumv\ has applications in several communication network design problems (see~\cite{wang1,wang2,n3} for further details).
The following computational complexity results are known regarding \uumv\ and \mse\
for a graph with $n$ nodes and $m$ edges (see~\cite{n1,n2}): 
\begin{itemize}
\item
\mse\ does not admit a 
$2^{\log ^{1-\eps} n}$-approximation for any constant $\eps>0$ unless 
$\NP\subseteq\,$DTIME$\left(n^{\log\log n}\right)$.
\smallskip
\item
\uumv\ admits a $\left\lfloor {{\kappa}}/ ({r+1}) \right\rfloor$-approximation. However, no non-trivial approximation 
of \uumv\ that depends on $m$ and/or $n$ only is currently known.

\smallskip

\item
\mse\ admits a 
$\min\left\{ \, n^{{3}/{4}}, \, m^{{1}/{2}} \,\right\}$-approximation.
\end{itemize}

\subsubsection{Greedy Fails for \uumv\ or \mse\ Even for Hyperbolic Graphs ({\em\IE}, Graphs With Constant $\delta$)}

Several routing problems have been looked at for hyperbolic graphs (\IE, constant $\delta$) in the 
literature before (\EG, see~\cite{Eppstein-Goodrich,K2007}) and, for these problems, it is often seen
that simple greedy strategies \emph{do} work.
However, that is unfortunately not the case with \uumv\ or \mse. For example, one obvious greedy strategy 
that can be designed is as follows.

\smallskip

\begin{center}
\begin{tabular}{p{0.9\textwidth}}
\toprule
(* Greedy strategy *)
\\
\midrule
{\bf Repeat} $\kappa$ times 
\\
\hspace*{0.2in}
Select a new path between $s$ and $t$ that shares a minimum number 
\\
\hspace*{0.4in}
of edges with the 
already selected paths
\\
\bottomrule
\end{tabular}
\end{center}

\begin{figure*}[htbp]
\hspace*{-0.2in}
\includegraphics{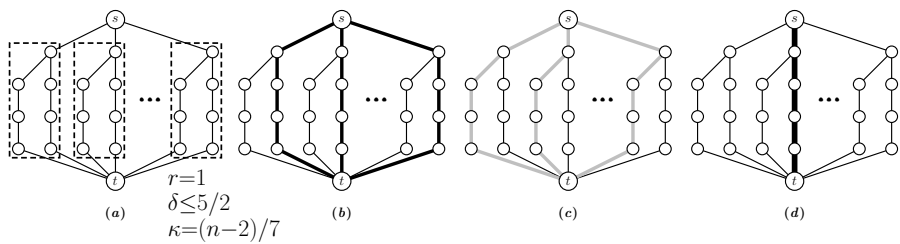}
\caption{\label{bad-ex1}A bad example for the obvious greedy strategy.
(\emph{a}) 
The given graph in which every node except $s$ and $t$ has degree at most $3$ and $\delta\leq 5/2$ . 
(\emph{b}) 
Greedy first selects the $(n-2)/14$ edge-disjoint shortest paths shown in thick black.
(\emph{c}) 
Greedy then selects the shortest paths shown in light gray one by one, each of which increases the number of shared edges by one more.
Thus, greedy uses $(n-2)/7$ shared edges.
(\emph{d}) 
An optimal solution uses only $5$ edges, \IE, 
$\OPT_{ \!\!\!\!\text{{\footnotesize\uumv}} \!\! }(G,s,t,1,{\kappa})=5$. 
}
\end{figure*}

The above greedy strategy can be arbitrarily bad even when $r=1$, $\delta\leq 5/2$ and every node except $s$ and $t$ 
has degree at most three as illustrated in \FI{bad-ex1}; even qualifying the greedy step by selecting a shortest path 
among those that increase the number of shared edges the least does not lead to a better solution.

\subsubsection{Improved Approximations for \uumv\ or \mse\ for $\delta$ Up To $o(\log n/\log d)$}

{\bf Note that in the following lemma 
$d$ is the maximum degree of any node ``except $s$, $t$ and any node within a distance of 
$35\,\delta$ of $s$''} (degrees of these nodes may be arbitrary).
For up to $\delta=o\left( {\log n}/ {\log d} \right)$, 
the lemma provides the first non-trivial approximation of \uumv\ 
as a function of $n$ only (independent of $\kappa$) and 
improves upon the currently best 
$\min\left\{ \, n^{{3}/{4}}, \, m^{{1}/{2}} \,\right\}$-approximation
of \mse\ for arbitrary graphs.

\begin{lemma}\label{equi2}
Let $d$ be the maximum degree of any node except $s$, $t$ and any node within a distance of 
$35\,\delta$ of $s$ (degrees of these nodes may be arbitrary).
Then, \uumv\ (and, consequently also \mse) for a $\delta$-hyperbolic graph $G$ can be approximated within a factor of
$O \left( \max \left\{ \log n, \, d^{\,O(\delta)} \right\} \, \right)$.
\end{lemma}

\begin{remark}
\label{equi2-com}
Thus for fixed $d$ Lemma~\ref{equi2} provides improved approximation as long as $\delta=o(\log n)$.
{\bf Note that our approximation ratio is independent of the value of $\kappa$}.
Also note that $\delta=\Omega(\log n)$ allows expander graphs as a sub-class of $\delta$-hyperbolic graphs for which 
\uumv\ or \mse\ is expected to be harder to approximate.
\end{remark}

\subsubsection*{Proof of Lemma~\ref{equi2}}

Our proof strategy has the following two steps:
\begin{enumerate}[label=$\blacktriangleright$]
\item
We define a new more general problem which we call the 
\emph{edge hitting set problem for size constrained cuts} (\ehssc), 
and show that 
\uumv\ (and thus \mse) has the {\em same approximability properties as} \ehssc\ by 
characterizing optimal solutions of \uumv\ in terms of optimal solutions of 
\ehssc.

\smallskip

\item
We then provide a suitable approximation algorithm for \ehssc.
\end{enumerate}
\begin{problem}[Edge hitting set for size-constrained cuts (\ehssc)]
The input to \ehssc\ is 
a graph $G=(V,E)$, two nodes $s,t\in V$, and a positive integer $0<k\leq |E|$. 
Define a size-constrained $s$-$t$ cut to be a $s$-$t$ cut $S$ such that the number of cut-edges $\cut_G(S,s,t)$  
is at most $k$.
The goal of \ehssc\ is to find a hitting set of minimum cardinality for all size-constrained $s$-$t$ cuts of $G$, \emph{\IE}, 
find $\widetilde{E}\subset E$ such that 
$|\,\widetilde{E}\,|$ is minimum and 
\[
\forall\, s\in S\subset V\setminus \{t\} \colon 
											    \left|\,\cE_G(S,s,t)\,\right|\leq k \, \Rightarrow \, 
											    \cE_G(S,s,t) \cap \widetilde{E} \neq \emptyset
\]
\end{problem}

We will use the 
notation $E_{  \!\!\text{{\footnotesize\ehssc}} \!\! }(G,s,t,k)$ 
to denote an optimal solution 
containing 
$\OPT_{ \!\!\!\!\text{{\footnotesize\ehssc}} \!\! }(G,s,t,k)$ edges
of an instance of \ehssc.

\begin{lemma}[Relating \ehssc\ to \uumv]\label{equi}

$
\OPT_{ \!\!\!\!\text{{\footnotesize\uumv}} \!\! }(G,s,t,r,\kappa)
=
\OPT_{ \!\!\!\!\text{{\footnotesize\ehssc}} \!\! }\left( G,s,t, \left\lceil {\kappa}/{r}\right\rceil -1 \right)
$.
\end{lemma}

\begin{proof}
Note that {\em any} feasible solution for \uumv\ {\em must} contain at least one edge from every
collection of cut-edges $\cE_G(S,s,t)$ satisfying
$\left|\,\cE_G(S,s,t)\,\right|\leq \left\lceil {\kappa}/{r}\right\rceil -1$, 
since otherwise the number of paths going from 
$\cE_G(S,s,t)$ to $V\setminus\cE_G(S,s,t)$ is at most 
$
r \left( \left\lceil \frac {\kappa} {r}\right\rceil -1 \right) \! < \! \kappa
$.
Thus we get
$
\OPT_{ \!\!\!\!\text{{\footnotesize\uumv}} \!\! }(G,s,t,r,\kappa)
\geq
\OPT_{ \!\!\!\!\text{{\footnotesize\ehssc}} \!\! }\left( G,s,t, \left\lceil {\kappa}/{r}\right\rceil -1 \right)
$. 

On the other hand, 
$
\OPT_{ \!\!\!\!\text{{\footnotesize\uumv}} \!\! }(G,s,t,r,\kappa)
\leq
\OPT_{ \!\!\!\!\text{{\footnotesize\ehssc}} \!\! }\left( G,s,t, \left\lceil {\kappa}/{r}\right\rceil -1 \right)
$ 
can be argued as follows.
Consider the set of edges 
$E_{  \!\!\text{{\footnotesize\ehssc}} \!\! }\left(G,s,t, \left\lceil {\kappa}/{r}\right\rceil -1 \right)$ in an optimal hitting set, 
and set the capacity $c(e)$ of every edge $e$ of $G$ as 
\[
c(e)=
\left\{
\begin{array}{ r l }
\infty, & \mbox{if $e\in E_{  \!\!\text{{\footnotesize\hssc}} \!\! } \left(G,s,t, \left\lceil {\kappa}/{r}\right\rceil -1 \right)$}
\\
r, & \mbox{otherwise}
\end{array}
\right.
\]
The value of the minimum $s$-$t$ cut for $G$ is then at least 
$\min \left\{ \infty, \, r \times \left\lceil {\kappa}/{r}\right\rceil \right\}\geq\kappa$ which implies 
(by the standard max-flow-min-cut theorem) 
the existence of $\kappa$ flows each of unit value. The paths taken by these $\kappa$ flows provide our desired $\kappa$ paths for \uumv.
Note that at most $r$ paths go through any edge $e$ with $c(e)\neq\infty$ and thus 
$
\OPT_{ \!\!\!\!\text{{\footnotesize\uumv}} \!\! }(G,s,t,r,\kappa)
\leq
\big| \, \{ e \,|\, c(e)\neq\infty  \} \,\big| 
=
\OPT_{ \!\!\!\!\text{{\footnotesize\ehssc}} \!\! }\left( G,s,t, \left\lceil {\kappa}/{r}\right\rceil -1 \right)
$.
\end{proof}

Now, we turn to providing a suitable approximation algorithm for \ehssc.
Of course, \ehssc\ has the following obvious \emph{exponential-size} $\LP$-relaxation since it is after all
a hitting set problem:

\begin{center}
\begin{tabular}{l}
\emph{minimize} $\sum_{e\in E}x_e\,\,\,$
\emph{subject to} 
\\
[2pt]
\hspace*{0.2in}
\begin{tabular}{r c l} 
$\forall\, s\in S\subset V\setminus \{t\}$ such that  
            $\cut_G(S,s,t)\leq k$ & $\!\!\colon\!\!$ & $\sum_{e\,\in\,\cE_G(S,s,t)} x_e\geq 1$
\\
$\forall\, e\in E$ & $\!\!\colon\!\!$ & $x_e\geq 0$
\end{tabular}
\end{tabular}
\end{center}

Intuitively, there are at least two reasons why such a $\LP$-relaxation may not be of sufficient interest.
Firstly, known results may imply a large integrality gap. Secondly, it is even not very clear if the $\LP$-relaxation
can be solved exactly in a time efficient manner.
Instead, we will exploit the hyperbolicity property and use 
Lemma~\ref{const-thm} to derive our approximation algorithm.

\begin{lemma}[Approximation algorithm for \ehssc]
\label{main4}
\ehssc\ admits a 
\\
$O \left( \max \left\{ \delta\,\log n, \, d^{\,O(\delta)} \right\} \, \right)$-approximation.
\end{lemma}

\begin{proof}
Our algorithm for \ehssc\ can be summarized as follows:

\smallskip

\begin{center}
\begin{tabular}{p{0.97\textwidth}}
\toprule
{Algorithm for \ehssc}
\\
\midrule
\hspace*{-0.1in}
{\bf If} $k\leq d^{\,12\delta+1}$ {\bf then} 
\\
\hspace*{0.0in} $\cA\leftarrow\emptyset$, $j\leftarrow 0$, set the capacity $c(e)$ of every edge $e$ to $1$
\\
\hspace*{0.0in} {\bf while} there exists a $s$-$t$ cut of capacity at most $k$ {\bf do}
\\
\hspace*{0.4in} $j\leftarrow j+1$, let $\cF_j$ be the edges of a $s$-$t$ cut of capacity at most $k$
\\
\hspace*{0.4in} $\cA\leftarrow\cA\cup\cF_j$, set $c(e)=\infty$ for every edge $e\in\cF_j$
\\
\hspace*{0.0in} {\bf return} $\cA$ as the solution
\\
\hspace*{-0.1in}
{\bf else} \hspace*{0.0in} $(*$ $k>d^{\,12\delta+1}$ $*)$
\\
\hspace*{0.0in} {\bf return} all the edges in a shortest path between $s$ and $t$ as the solution $\cA$
\\
\bottomrule
\end{tabular}
\end{center}

The following case analysis of the algorithm shows the desired approximation bound.

\medskip
\noindent
{\bf Case 1: $k\leq d^{\,12\delta+1}$}.
Let $\cF_1,\cF_2,\dots,\cF_\ell$ be the sets whose edges were added to $\cA$; thus, 
$|\cA|\leq k\,\ell$.
Since $|\cF_j|\leq k$ and $\cF_j\,\cap\,\cF_{j'}=\emptyset$ for $j\neq j'$, 
{
\def\OldComma{,}
    \catcode`\,=13
    \def,{%
      \ifmmode%
        \OldComma\discretionary{}{}{}%
      \else%
        \OldComma%
      \fi%
    }%
$\OPT_{ \!\!\!\!\text{{\footnotesize\ehssc}} \!\! }(G,s,t,k)\geq \ell$, 
}
thus providing 
an approximation bound of 
$k\leq d^{\,12\delta+1}$.

\medskip
\noindent
{\bf Case 2: $k>d^{\,12\delta+1}$ and $\dist_G(s,t)\leq 48\,\delta \,+\, 8\,\delta\,\log n$}. 
Since 
{
\def\OldComma{,}
    \catcode`\,=13
    \def,{%
      \ifmmode%
        \OldComma\discretionary{}{}{}%
      \else%
        \OldComma%
      \fi%
    }%
$\OPT_{ \!\!\!\!\text{{\footnotesize\ehssc}} \!\! }(G,s,t,k)\geq 1$, 
}
this provides a 
$O( \delta\,\log n)$-approximation.

\medskip
\noindent
{\bf Case 3: $k>d^{\,12\delta+1}$ and $\dist_G(s,t)>48\,\delta + 8\,\delta\,\log n$}. 
Use Lemma~\ref{const-thm} to find a collection $S_1,S_2,\dots,S_\ell$ of 
$\ell= ( {\dist_G(s,t)-8\,\delta\,\log n} ) / ( {50\,\delta})$ edge and node disjoint $s$-$t$ cuts. 
Since 
$\cut_G\left( S_j,s,t \right)\leq d^{\,12\delta+1}<k$, any valid solution of \ehssc\ 
must select \emph{at least} one edge from 
$\cE_G\left( S_j,s,t \right)$.
Since the cuts are edge and node disjoint, it follows that 
\[
\OPT_{ \!\!\!\!\text{{\footnotesize\ehssc}} \!\! }(G,s,t,k)\geq 
\left( {\dist_G(s,t)-8\,\delta\,\log n} \right) / ({50\,\delta}) 
\]
Since we return all the edges in a shortest path between $s$ and $t$ as the solution, 
the approximation ratio achieved is
$
{\dist_G(s,t) } / 
{
\left( \frac{\dist_G(s,t)-8\, \delta\, \log n}{50\, \delta} \right)
}
<
100\delta
$.
\end{proof}

\subsection{Application to the Small Set Expansion Problem}
\label{sec-small-expansion}

The \emph{small set expansion} (\sse) problem was studied by Arora, Barak and Steurer in~\cite{ABS10} (and also by several 
other researchers such as~\cite{bansal,exp-ref1,exp-ref2,exp-ref3,exp-ref4}) 
in an attempt to understand the computational difficulties surrounding the 
Unique Games Conjecture (UGC). 
To define \sse, we will also use the \emph{normalized} edge-expansion of a graph 
which is defined as follows~\cite{chung97}. 
For a subset of nodes $S$ of a graph $G$, let $\vol_G(S)$ denote the sum of degrees of the 
nodes in $G$. Then, the normalized edge expansion ratio $\Phi_G(S)$ of a subset $S$ of nodes of at most 
${|V|}/{2}$ nodes of $G$ is defined as 
$\Phi_G(S)={\cut_G(S)}/{\vol_G(S)}$.
Since we will deal with only $d$-regular graphs in this subsection, $\Phi_G(S)$ will simplify to 
${\cut_G(S)}/ ( {d\,|S|} \, )$.

\begin{definition}[(\sse\ Problem]\label{def-sse}
{\bf [$\,$a case of~\cite[Theorem 2.1]{ABS10}, rewritten as a problem$\,$]}
Suppose that we are given a $d$-regular graph $G=(V,E)$ for some fixed $d$, and suppose $G$ has a subset of 
at most $\zeta n$ nodes $S$, for some constant $0<\zeta<{1}/{2}$, 
such that $\Phi_G(S)\leq\eps$ 
for some constant $0<\eps\leq 1$. 
Then, find as efficiently as possible a subset $S'$ of at most $\zeta n$ nodes such that 
$\Phi_G(S)\leq\eta\,\eps$ 
for some ``universal constant'' $\eta>0$. 
\end{definition}

In general, computing a very good approximation of the \sse\ problem seems to be quite hard; the approximation ratio of 
the algorithm presented in~\cite{exp-ref4} roughly deteriorates proportional to $\sqrt{\log ({1}/{\zeta})}$, and 
a $O(1)$-approximation described in~\cite{bansal} works only if the graph excludes two specific minors. 
The authors in~\cite{ABS10} showed how to design a sub-exponential time (\IE, $O\left(2^{\,c\,n}\right)$ time
for some constant $c<1$) algorithm for the above problem.
As they remark, expander like graphs are somewhat easier instances of \sse\
for their algorithm, and it takes some 
non-trivial technical effort to handle the ``non-expander'' graphs. Note that {\em the class of $\delta$-hyperbolic graphs for}
$\delta=o(\log n)$ {\em is a non-trivial proper subclass of non-expander graphs}.
We show that \sse\ (as defined in Definition~\ref{def-sse}) can be solved in polynomial time for such a proper subclass of
non-expanders.

\begin{lemma}\label{thm-sse-hyper}
\emph{(}{\bf polynomial time solution of \sse\ for $\delta$-hyperbolic graphs when $\delta$ is sub-logarithmic and 
$d$ is sub-linear}\emph{)}
Suppose that $G$ is a $d$-regular $\delta$-hyperbolic graph.
Then the \sse\ problem for $G$ can be solved in polynomial time provided $d$ and $\delta$ satisfy: 
\[
d \leq 2^{\,\log^{(1/3)-\rho} n} 
\,\,\,
\text{and}
\,\,\,
\delta \leq \log^{\rho} n
\,\,\,
\text{for some constant $0<\rho<1/3$}
\]
\end{lemma}

\begin{remark}
Computing the minimum node expansion ratio of a graph is in general $\NP$-hard and 
is in fact \sse-hard to approximate within a ratio of $C\,\sqrt{\h_G\,\log d}$ for some constant $C>0$~\cite{LRV13}.
Since we show that \sse\ is polynomial-time solvable for $\delta$-hyperbolic graphs
for some parameter ranges, the hardness result of~\cite{LRV13} does not
directly apply for graph classes that belong to these cases, and thus additional arguments may be 
needed to establish similar hardness results
for these classes of graphs.
\end{remark}

\begin{proof}
Our proof is quite similar to that used for Theorems~\ref{main-nested}.
But, instead of looking for 
smallest possible non-expansion bounds, we now relax the search and allow us to consider subsets of nodes 
whose expansion is just enough to satisfy the requirement. This relaxation helps us to ensure the size requirement of the subset
we need to find.

We will use the construction in the proof of Theorem~\ref{main-nested} in this proof, 
so \emph{we urge the readers to familiarize themselves 
with the details of that proof before reading the current proof}.
Note that 
$h_G(S)\leq\eps$ implies $\Phi_G(S)\leq {d\,h_G(S)}/{d}\leq\eps$.
We select the nodes $p$ and $q$ such that $\Delta=\dist_G(p,q)=\log_d n = {\log n} / {\log d}$, and 
set $\mu=1/2$.
Note that 
$ 
({ 360 \log n } ) / \left( { \Delta\,2^{\,   {\Delta^{\mu}} / ( {28\,\delta\,\log(2d) } )     } } \right)
<
\left({1}/{\Delta}\right)^{1-\mu}
$ 
since
\begin{multline*}
({ 360 \log n } ) / \left( { \Delta\,2^{\,   {\Delta^{\mu}} / ( {28\,\delta\,\log(2d) } )     } } \right)
<
\left({1}/{\Delta}\right)^{1-\mu}
\\
\Leftarrow\,
( { 360 \log d } ) / \left( { 2^{\, { \left( \log n \right)^{1/2} } / \left( {56\,\delta\,\left( \log d \right)^{3/2} } \right) } } \right)
<
\left({\log d}/{\log n}\right)^{1/2}
\\
\Leftarrow \,
9+ \log\log n/2 
<
\left ( { \left( \log n \right)^{1/2} } \right) / \left( {56\,\log^{(1-\rho)/2} n }  \right)
-
\log^{(1-\rho)/2} n
\end{multline*}
and the last inequality clearly holds for sufficiently large $n$.

First, suppose that there exists $0\leq r \leq \frac{\Delta}{3}-\alpha\Delta$ such that
$h_G\! \left(  \cB_{G_{-\cC}}(p,r) \right)
\!=\!
h_G\left(  \cB_{G}(p,r) \right)
\leq\eps$. 
We return 
$S'=\cB_{G}(p,r)$ as our solution, To verify the size requirement, note that 
\begin{multline}
\left| \,\cB_{G}(p,r)\, \right|
\leq 
\left| \,\cB_{G}\left(p, ({\Delta}/{3})-\alpha\,\Delta \right)\, \right| 
<
\left| \,\cB_{G}\left(p, {\Delta}/{3}\right)\, \right| 
\\
<
\sum\limits_{i=0}^{{\Delta}/{3}} d^{\,i}
< 
d^{({\Delta}/{3})+1}
=d\,n^{1/3}
<\zeta\,n
\label{eq:size}
\end{multline}
where the last inequality follows since 
$
d \leq 2^{\,\log^{(1/3)-\rho} n} 
$ 
and $\zeta$ is a constant.

Otherwise, no such $r$ exists, and this implies 
\begin{multline*}
\left| \,\cB_{G}\left(p, ({\Delta}/{3})-\alpha\,\Delta \right)\, \right| 
\geq
\left( 1+\eps \right)^{ ({\Delta}/{3})-\alpha\,\Delta }
\\
>
\left( 1+\eps \right)^{ {\Delta}/{4} }
\geq
\bee^{ {\eps \Delta}/{8} }
=
\bee^{ {\eps \log_d n}/{8} }
=
n^{ {\eps \log_d \bee}/{8} }
\end{multline*}
Now there are two major cases as follows.

\bigskip
\noindent
{\bf Case 1: there exists at least one path between $p$ and $q$ in $G_{-\cC}$}.

\medskip
We know that $\dist_{G_{-\cC}}(p,q)\geq ({\Delta}/{60})2^{\,\alpha\,\Delta/\delta}$
and (by choice of $p$) 
\\
$
\left|\,
\cB_{G_{-\cC}} \left( p, {\dist_{G_{-\cC}}(p,q)}/{2} \right)
\,\right|
<{n}/{2}
$.
Let $p=u_0,u_1,\dots,u_{t-1},u_t=q$ be the nodes in successive order on a shortest path 
from $p$ to $q$ of length $t=\dist_{G_{-\cC}}(p,q)$.
Perform a BFS starting from $p$ in $G_{-\cC}$, and let $\cL_i$ be the sets of nodes 
at the $i\tx$ level (\IE, $\forall\,u\in\cL_i\colon \dist_{G_{-\cC}}\left(p,u\right)=i$).
Note that $\big| \, \bigcup_{j=0}^{\,t/2} \cL_j \,\big|\leq {n}/{2}$.
Consider the levels $\cL_0,\cL_1,\dots,\cL_{t/2}$, and partition the ordered sequence of 
integers $0,1,2,\dots,{t}/{2}$ into consecutive blocks $\Delta_0,\Delta_1,\dots,\Delta_{\left(1+({t}/{2})\,\right)/\kappa-1}$
each of length $\eta=({8}/{\eps})\,\ln n$, \IE,
\[
\underbrace{0,1,\dots,\eta-1}_{\Delta_0},\,
\underbrace{\eta,\eta+1,\dots,2\eta-1}_{\Delta_1},
\dots\dots,
\underbrace{({t}/{2})-\eta+1,({t}/{2})-\eta+2,\dots,({t}/{2})}_{\Delta_{\left(1+({t}/{2})\,\right)/\eta -1}}
\]
We claim that for every $\Delta_i$, there exists an index $i^*$ within $\Delta_i$ (\IE, there exists an index 
$i\,\eta \leq i^* \leq i\,\eta+\eta-1$) such that
$h_G \left( \cL_{i^*} \right) \leq\eps$.
Suppose for the sake of contradiction that this is not true. 
Then, it follows that
\begin{multline*}
\forall\,i\,\eta \leq j \leq i\,\eta+\eta-1 \,\colon\,
h_{G_{-\cC}} \left( \cL_j \right)
\geq 
{h_G \left( \cL_j \right)}/{2}
>
{\eps}/{2}
\\
\Rightarrow\,
\left| \cL_{i\,\eta+\eta-1} \right|
>
\left|\, \cL_{i\,\eta} \,\right|
\left( 1 + ({\eps}/{2})\, \right)^{\eta}
\geq 
\left( 1 + ({\eps}/{2}) \, \right)^{ ({8}/{\eps})\,\ln n}
\\
\geq
\bee^{ \left({\eps}/{4} \right)\,\left( ({8}/{\eps})\,\ln n \right) }
=
n^2
>n
\end{multline*}
which contradicts the fact that 
$\big| \, \bigcup_{j=0}^{\,t/2} \cL_j \,\big|\leq {n}/{2}$.
Since 
$
\hspace*{-0.6in}
\sum\limits_{i=0}^{\hspace*{0.6in}\left(1+ ({t}/{2}) \, \right)/\kappa -1 } 
\hspace*{-0.6in} \left|\, \cL_{i^*}\,\right|
<{n}/{2}
$,
there exists a set $\cL_{k^*}$ such that 
$h_G \left( \cL_{k^*} \right) \leq\eps$ and 
\begin{multline*}
\left|\, \cL_{k^*} \,\right| 
<
\frac{ {n}/{2} }
{
\left( {1+ ({t}/{2}}) \, \right) / {\kappa}
}
<
{n \kappa}/{t}
<
( {8\,n\, \ln n} ) / \left( {\eps \,({\Delta}/{60})\, 2^{ { {\Delta}^{1/2} } / ( { 7\,\delta\,\log (2d) } ) \, } } \right)
\\
\leq
 \left( {480\,n\, {\log^{(1/3)-\rho} n}  }   \right) 
 /
 \left( { \eps\, 2^{ { \, ( {\log^{\rho/2} n} ) } / { 14 } } } \right)
<
\zeta\,n
\end{multline*}

\bigskip
\noindent
{\bf Case 2: there is no path between $p$ and $q$ in $G_{-\cC}$}.

\medskip
In this case, we return 
$
\cB_{G_{-\cC}}\left(p, ({\Delta}/{3})-\alpha\,\Delta \right) 
=
\cB_{G}\left(p, ({\Delta}/{3})-\alpha\,\Delta \right) 
$ as our solution.
The size requirement follows since 
$\left| \,\cB_{G}\left(p, ({\Delta}/{3})-\alpha\,\Delta \right)\, \right| <\zeta\,n$
was shown in \eqref{eq:size}.
Note that nodes in 
$\cB_{G}\left(p, ({\Delta}/{3}) -\alpha\,\Delta \right)$ can only be connected to nodes in $\cC$, and thus
\begin{multline*}
h_G \left( \cB_{G}\left(p, ({\Delta}/{3}) -\alpha\,\Delta \right) \,\right)
\leq
 { |\, \cC \,|} 
 \,/\,
 {
\left|\, \cB_{G}\left(p, ({\Delta}/{3}) -\alpha\,\Delta \right) \,\right|
}
\\
\leq
\left ( \, ({\Delta}/{3}) d^{\,\alpha \Delta} \right)
 /
\left( n^{ {\eps \log_d \bee} / {8} } \right)
<
n^{\alpha - ( {\eps \log_d \bee} / {8} ) }
\log n
\\
<
n^{1/(7 \, \Delta^{1/2} \log (2d) )  \,-\, ( {\eps } / ( {8 \ln d}) \,  ) }
\log n
<
\eps
\end{multline*}
where the penultimate inequality follows since $\Delta=\omega(1)$.

In all cases, the desired subset of nodes can be found in $O\left(n^2\log n\right)$ time.
\end{proof}

\section{Conclusion and Open Problems}
\label{sec-concl}

In this paper we have provided the first known non-trivial bounds on expansions and cut-sizes for graphs
as a function of the hyperbolicity measure $\delta$, and have shown how these bounds and their related proof
techniques lead to improved algorithms for two related combinatorial problems. We hope that these results
will stimulate further research in characterizing the computational complexities of related combinatorial problems 
over asymptotic ranges of $\delta$. In addition to the usual future research of improving our bounds, the
following interesting research questions remain: 
\begin{enumerate}[label=$\blacktriangleright$]
\item
Can one use Lemma~\ref{thm-sse-hyper} or similar results to get a polynomial-time solution of UGC for some asymptotic
ranges of $\delta$? An obvious recursive application using the approach in~\cite{ABS10} encounters a hurdle since 
hyperbolicity is not a hereditary property (cf.\ Section~\ref{sec-topo}), \IE, removal of nodes or edges may change $\delta$ sharply;
however, it is conceivable that a more clever approach may succeed.

\item
Can our bounds on expansions and cut-sizes be used to get an improved approximation for the multicut
problem~\cite[Problem 18.1]{V01} provided $\delta=o(\log n)$?
\end{enumerate}

\begin{acknowledgements}
B. DasGupta, N. Mobasheri and F. Yahyanejad
thankfully acknowledge supported from NSF grant IIS-1160995 for this research.
M. Karpinski was supported in part by DFG grants.
The problem of investigating expansion properties of $\delta$-hyperbolic graphs was raised originally to some of the 
authors by A. Wigderson.
\end{acknowledgements}

\end{document}